\documentclass[journal]{IEEEtran}
\usepackage{amsmath}
\usepackage{amssymb}
\usepackage{amsthm}
\usepackage{algorithm}
\usepackage{algorithmic}
\usepackage{color}
\usepackage{xcolor}
\usepackage{caption}
\usepackage{epsfig,latexsym}
\usepackage{float}
\usepackage{fancyhdr}
\usepackage{graphicx}
\usepackage{indentfirst}
\usepackage{mdwmath}
\usepackage{mdwtab}
\usepackage{subfigure}
\usepackage{setspace}
\usepackage{stfloats}
\usepackage{times}
\usepackage{url}
\usepackage{multirow}
\usepackage{booktabs}

\usepackage[colorlinks=true, allcolors=blue]{hyperref}

\newtheorem{remark}{Remark}
\newtheorem{theorem}{Theorem}

\newtheorem{lemma}{Lemma}

\newtheorem{corollary}{Corollary}

\newtheorem{proposition}{Proposition}

\begin{document}

\title{On the Performance of Fluid Antenna Systems under Block-Diagonal Correlation Model}
\author{Jiangsheng Huangfu, Zhengyu Song,~\IEEEmembership{Member,~IEEE}, \\
Tianwei Hou,~\IEEEmembership{Member,~IEEE}, Anna Li,~\IEEEmembership{Member,~IEEE}, and Kai-Kit Wong,~\IEEEmembership{Fellow,~IEEE}.
 \thanks{This work was supported in part by the Beijing Natural Science Foundation L232041. ({\em Corresponding author: Tianwei Hou}).}
\thanks{Jiangsheng Huangfu and Zhengyu Song are with the School of Electronic and Information Engineering, Beijing Jiaotong University, Beijing 100044, China (e-mail: 25110117@bjtu.edu.cn; songzy@bjtu.edu.cn; twhou@bjtu.edu.cn).}
\thanks{Tianwei Hou is with the Beijing Key Laboratory of Convergent Communications and Networking Technologies in LEO Satellite Systems, and also with the State Key Laboratory of Networking and Switching Technology, Beijing University of Posts and Telecommunications, Beijing 100876, China (email: htw@bupt.edu.cn).}
\thanks{Anna Li is with the School of Computing and Communications, Lancaster University, Lancaster LA1 4WA, U.K. (e-mail: a.li16@lancaster.ac.uk).}
\thanks{K. K. Wong is affiliated with the Department of Electronic and Electrical Engineering, University College London, Torrington Place, WC1E 7JE, United Kingdom and he is also affiliated with the Department of Electronic Engineering, Kyung Hee University, Yongin-si, Gyeonggi-do 17104, Korea (e-mail: kai-kit.wong@ucl.ac.uk).}
 }

\vspace{-8mm}

\maketitle

\begin{abstract}
Fluid antenna systems (FASs) have recently emerged as a promising reconfigurable antenna technology for future wireless networks, owing to their unique ability to exploit fine-grained spatial channel variations within a compact aperture. In this paper, a single-input multiple-output (SIMO) FAS employing maximum-ratio combining (MRC) is investigated under the block-diagonal correlation model, where the ports of FAS are partitioned into independent blocks and the strongest port within each block is selected for MRC combining. Exact outage probability (OP) expressions are first derived in both convolution and characteristic-function forms. To gain further insights, closed-form high-SNR asymptotic expressions are developed, from which the diversity order is shown to approximate the number of ports. This result reveals that block partitioning influences only the coding gain and can therefore be optimized without compromising the diversity performance. For the ergodic rate (ER), a Gamma-matching approximation together with a tighter Jensen-based approximation is derived in closed form. Simulation results corroborate the analytical framework and demonstrate that: i) increasing either the number of ports or the number of blocks improves the system performance; ii) the diversity order depends solely on the number of ports; and iii) the proposed SIMO-FAS achieves comparable or superior outage performance to conventional MRC receivers despite employing fewer combining branches.
\end{abstract}


\section{Introduction}

The forthcoming sixth-generation (6G) wireless systems are expected to support ultra-high data rates, massive connectivity, and stringent reliability requirements \cite{6G1,6G2,6G3,6G4}. To meet these demands, a series of breakthrough physical-layer technologies have been proposed and widely investigated. Among them, fluid antenna systems (FASs) have emerged as a promising reconfigurable-antenna paradigm for 6G, attracting growing research interest \textcolor{blue}{\cite{fas1,fas2,fas3}}.

Unlike traditional antenna systems (TASs) with fixed antenna positions, FAS represents a broad class of antenna architectures that can alter their shapes, positions, or other physical configurations to reconfigure their gain pattern, radiation characteristics, and spatial response \cite{fas4,fas,new2025fluid}. By exploiting this reconfigurability, FAS introduces new degrees of freedom in the spatial domain, enabling the system to adaptively maximize signal strength by selecting the most favorable port locations within a compact aperture. 
Owing to its flexibility, FAS can be deployed at either the base station (BS) or the user equipment (UE), or at both ends simultaneously, making it readily compatible with existing wireless technologies such as reconfigurable intelligent surfaces (RIS) \cite{RIS,RIS1,RIS2} and non-orthogonal multiple access (NOMA) \cite{NOMA1,NOMA2,NOMA3,NOMA4}. Furthermore, the unique spatial reconfigurability of FAS has also inspired new communication paradigms, among which fluid antenna multiple access (FAMA) has recently emerged as a promising research direction \cite{hong2026fluid}.

The core advantage of FAS lies in its flexibility to reconfigure spatial diversity gains without expanding the physical size of the antenna array \cite{tutorial,new2024tutorial}. To realize this reconfigurability, several practical implementation approaches have been explored, including mechanically movable antennas driven by precision actuators \cite{mech1,mech2,mech3,mech4}, liquid-based antennas that exploit the deformability of conductive fluids \cite{liq1,liq2,liq3}, and pixel-reconfigurable antennas realized via electronically switchable radiating elements \cite{pixel1,pixel2,pixel3}. The design principles and structural trade-offs of these implementations have been comprehensively reviewed in \cite{Tong2025}. Moreover, several FAS prototypes have already been developed and experimentally demonstrated, validating the practical feasibility of this technology \cite{pro1,pro2,pro3,pro4}.

Due to the close spacing of ports within a compact aperture, the channel coefficients among neighboring ports exhibit significant spatial correlation, which makes performance analysis particularly challenging. Motivated by this, substantial efforts have been devoted to FAS channel modeling and performance characterization.
Early works on FAS performance analysis were largely focused on single-input single-output (SISO) settings with FAS deployed only at the UE side, relying on a simple yet insightful parametric channel model \cite{tutorial}. Under this framework, the outage probability (OP) under Rayleigh fading was first studied in \cite{fas}, and the ergodic capacity limits were characterized in \cite{FAS_ER}. Extending to more general fading environments, \cite{nakagami1} investigated the OP under Nakagami-$m$ fading, and \cite{nakagami2} provided a closed-form approximate expression under the same model. Exact expressions together with closed-form upper and lower bounds for the OP, as well as the ergodic rate (ER), were derived under Rician fading in \cite{rician}. The system performance under $\alpha$-$\mu$ fading was further examined in \cite{Alvim2024}.

A key difficulty in analyzing FAS performance stems from the need to obtain the multivariate joint cumulative distribution function (CDF) of channel gains over all available ports. However, such an exact characterization generally leads to high-dimensional nested integrals, rendering the resulting expressions analytically intractable. Early works resorted to a simple parametric channel model \cite{Beaulieu2011}, which failed to accurately reproduce the spatial correlation structure prescribed by Jakes' model, leading to overly optimistic performance predictions. To address this limitation, a two-stage approximation method was proposed in \cite{Khammassi2023} to more faithfully characterize the inter-port correlation under Jakes' model, achieving improved accuracy while retaining analytical tractability. Along a parallel line, the Gaussian copula framework \cite{copula1, copula2, copula3} has been adopted to model the joint distribution of correlated FAS channel envelopes, enabling closed-form derivations of performance metrics. Building upon these foundations, a block-diagonal correlation model was proposed in \cite{block}, which captures the dominant correlation structure of FAS channels by grouping ports into independent blocks, each governed by a uniform intra-block correlation coefficient. This model strikes an effective balance between modeling fidelity and mathematical tractability, as the block structure renders the joint distribution of port gains separable, thereby circumventing the intractability of high-dimensional integration that plagues exact analysis. The effectiveness of this model has since been validated in a variety of system configurations \cite{block1, block2, block3}.

Despite these advances, most existing works on FAS performance analysis have focused on configurations where only the UE side is equipped with FAS and a single port is activated for transmission, i.e., the SISO-FAS setting. However, practical FAS deployments may benefit from simultaneously activating multiple ports to exploit spatial diversity more aggressively, particularly in scenarios where a single port may not suffice to guarantee reliability requirements. The performance of such multi-port activation schemes remains relatively underexplored. Furthermore, while the block-diagonal correlation model has demonstrated its effectiveness in FAS performance analysis, works employing this model are still limited, and its application to the multi-port activation setting has received limited attention.

Motivated by these observations, this work considers a FAS-enabled SIMO system, in which the UE employs a FAS capable of simultaneously activating multiple ports. Based on the block-diagonal correlation model, we develop a tractable analytical framework, from which exact expressions for the OP and ER of the considered SIMO-FAS system are derived. Asymptotic analyses in the high-SNR regime are further conducted to unveil the achievable diversity gain and multiplexing gain, yielding both rigorous performance benchmarks and explicit design insights.

The major technical contributions of this work are summarized as follows.
\begin{itemize}
  \item We consider a SIMO system in which a BS equipped with a fixed-position antenna communicates with a UE equipped with a FAS with multiple activated ports. By adopting the block-diagonal correlation model that partitions the ports into independent blocks, an exact integral expression for the CDF of the selected-port SNR within each block is derived.

  \item Leveraging the statistical independence across different blocks, exact OP expressions for the considered SIMO-FAS are obtained in both convolution and characteristic-function forms. To gain deeper insights, closed-form high-SNR asymptotic expressions for the OP are derived. A Gamma-matching approximation is further developed. Based on the asymptotic results, the diversity order is obtained and shown to depend solely on the number of ports, irrespective of the number of blocks.

  \item For the ER analysis, we first derive an exact expression, based on which a Gamma-matching approximation and a tighter Jensen-based approximation are obtained. For the special case where each block contains a single port, exact closed-form expressions for both the OP and the ER are derived.

  \item Simulation results confirm that: 1) increasing either the number of ports or the number of blocks reduces the OP, while also improving the ER; 2) the diversity order depends solely on the number of ports; and 3) the proposed SIMO-FAS achieves comparable or superior reliability to conventional MRC receivers despite using fewer combining branches.
\end{itemize}

\subsection{Organization and Notations}

This paper proceeds as follows. Section II presents the underlying system configuration together with the adopted spatial correlation model. In Section III, the OP and ER are characterized through exact analysis and high-SNR asymptotic approximations. Section IV reports numerical results and examines the impact of key system parameters. The conclusion is given in Section V.

The major symbols and their definitions are summarized in Table~\ref{tab:notation}. \vspace{-1mm}

\begin{table}[h!]
\centering
\caption{Major notations summary.}
\resizebox{0.48\textwidth}{!}{%
\begin{tabular}{lc}
\toprule
Symbol & Definition \\
\midrule
\(\mathbb{E}\{\cdot\}\), \(\Pr(\cdot)\) & Expectation and probability operators \\
\(\Re(\cdot)\), \(\Im(\cdot)\) & Real and imaginary parts \\
\(|\cdot|\) & Magnitude or determinant \\
\(\|\cdot\|\) & Euclidean norm \\
\((\cdot)^T\), \((\cdot)^H\) & Transpose and conjugate transpose \\
\(\mathbf{I}_N\) & \(N \times N\) identity matrix \\
\(\mathbf{0}\) & Zero matrix or vector \\
\(\mathcal{CN}(\boldsymbol{\mu}, \mathbf{\Sigma})\) & Complex Gaussian distribution \\
\(\mathcal{N}(\mu, \sigma^2)\) & Real Gaussian distribution \\
\(\sim\) & Distributed as \\
\(\Gamma(\cdot)\) & Gamma function \\
\(J_0(\cdot)\) & Zeroth-order Bessel function of the first kind\\
\(Q_1(\cdot,\cdot)\) & First-order Marcum \(Q\)-function \\
\(E_n(\cdot)\) & Generalized exponential integral \\
\bottomrule
\end{tabular}%
}
\label{tab:notation}
\end{table}

\section{System and Channel Model}\label{sec:model}

\begin{figure}[t!]
\centering
\includegraphics[width=3.5in]{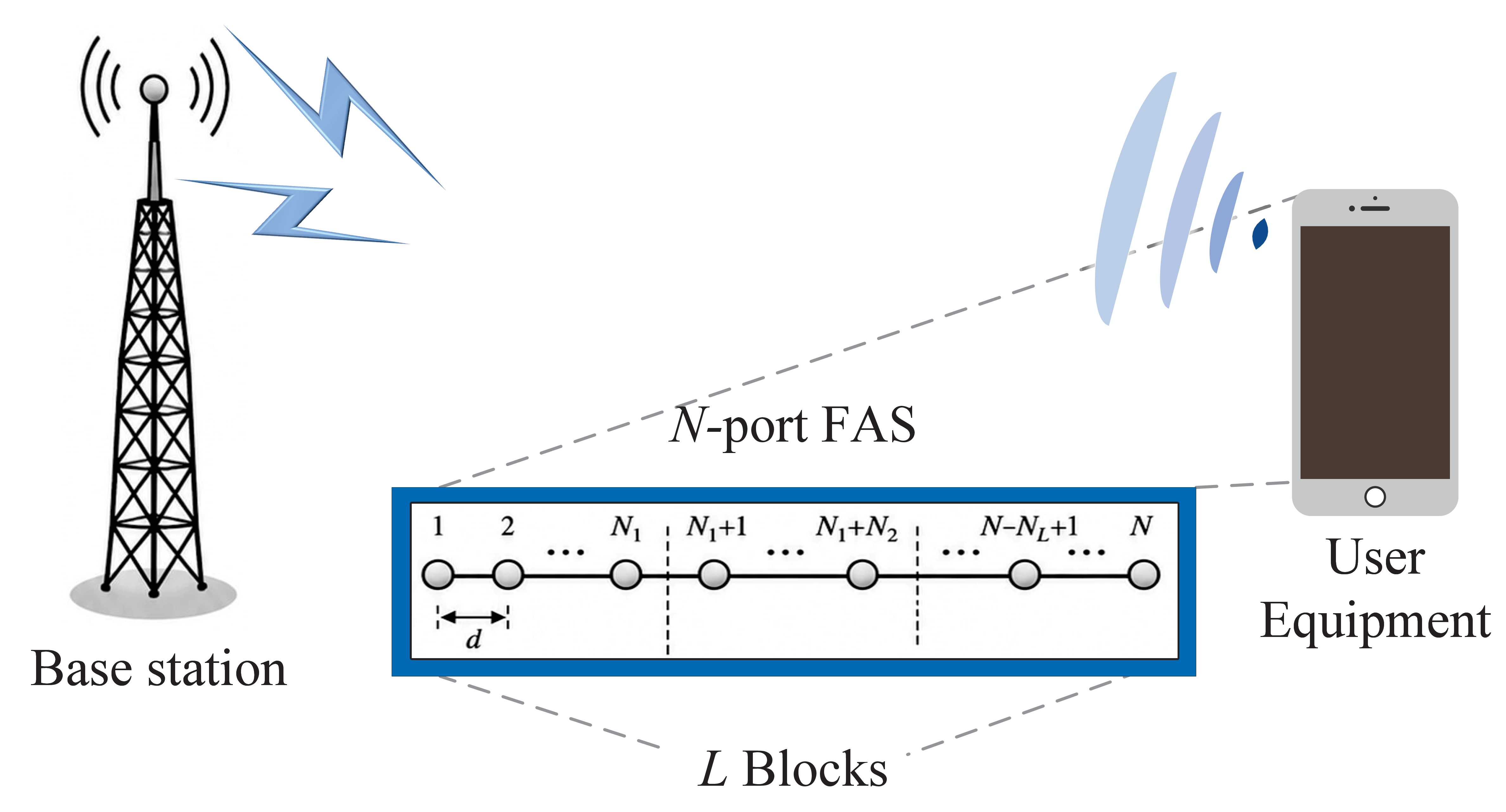}
\caption{SIMO-FAS under block-diagonal correlation model.}\label{sys}
\vspace{-3mm}
\end{figure}

This section introduces the system and channel model of the SIMO-FAS. We consider a communication scenario where the BS is equipped with a single fixed-position transmit antenna and the UE is equipped with a FAS comprising $N$ ports. The FAS aperture has a physical size of $W\lambda$, where $W$ denotes the normalized size of FAS and $\lambda$ is the wavelength. At the receiver, $L$ ports are simultaneously activated, and the BS transmits a single data stream to these $L$ ports.

Due to the dense spacing among ports within the FAS aperture, significant inter-port correlation arises. To characterize this correlation, we adopt the widely used Jakes' model \cite{fas}, under which the spatial correlation matrix of all $N$ ports is given by
\begin{equation}\label{corre1}
\mathbf{\Sigma}  \in  {\mathbb{R}}^{N \times  N}=
\begin{bmatrix}
\mu_{0}      & \mu_{1}      & \mu_{2}      & \cdots & \mu_{N-1} \\
\mu_{-1}     & \mu_{0}      & \mu_{1}      & \cdots & \mu_{N-2} \\
\vdots       & \vdots       & \ddots       & \ddots & \vdots \\
\mu_{-(N-1)} & \mu_{-(N-2)} & \cdots       & \mu_{-1} & \mu_{0}
\end{bmatrix},
\end{equation}
where $\mu_{n}$ denotes the spatial correlation coefficient, which is given by
\begin{equation}\label{corre2}
{\mu _n} = {J_0}\left( {\frac{{2\pi nW}}{{N - 1}}} \right) \equiv {J_0}\left( {2\pi nd} \right),
\end{equation}
where $J_0(\cdot)$ is the zeroth-order Bessel function of the first kind, and $d = \frac{W}{{N - 1}}$ denotes the distance between two neighboring ports. It can be observed that the resulting full correlation matrix $\mathbf{\Sigma}$ is a Toeplitz matrix, whose entries depend only on the relative port separation.

While the Jakes' correlation model accurately captures the physical inter-port correlation structure, the resulting full Toeplitz correlation matrix renders exact performance analysis intractable, particularly when multiple ports are simultaneously activated. To facilitate analytical tractability while preserving the essential correlation structure, we adopt the block-diagonal correlation model proposed in \cite{block}, in which the spatial correlation matrix $\mathbf{\Sigma}$ is approximated by a block-diagonal structure. 

As illustrated in Fig.~\ref{sys}, the $N$ ports are partitioned into $L$ blocks, where the $l$-th block contains $N_l$ ports with $\sum_{l=1}^{L} N_l = N$. Ports within the same block are modeled as correlated, while ports belonging to different blocks are treated as mutually independent. Specifically, for each of the $L$ blocks, the port achieving the maximum instantaneous SNR is selected, after which the selected ports from all blocks are combined through MRC at the receiver.

Under the block-diagonal correlation model, the approximate correlation matrix $\hat{\mathbf{\Sigma}}$ takes a block-diagonal form, given by
\begin{equation}\label{block1}
\mathbf{\hat{\Sigma}} \in \mathbb{R}^{N \times N}
=
\begin{bmatrix}
\mathbf{A}_{1} & \mathbf{0} & \cdots & \mathbf{0} \\
\mathbf{0} & \mathbf{A}_{2} & \cdots & \mathbf{0} \\
\vdots & \vdots & \ddots & \vdots \\
\mathbf{0} & \mathbf{0} & \cdots & \mathbf{A}_{L}
\end{bmatrix},
\end{equation}
where $\mathbf{A}_l \in \mathbb{R}^{N_l \times N_l}$ denotes the intra-block correlation matrix of the $l$-th block, which is modeled as an equicorrelation matrix of the form
\begin{equation}\label{block2}
\mathbf{A}_l =
\begin{bmatrix}
1 & \mu_l^2 & \cdots & \mu_l^2 \\
\mu_l^2 & 1 & \cdots & \mu_l^2 \\
\vdots & \vdots & \ddots & \vdots \\
\mu_l^2 & \mu_l^2 & \cdots & 1
\end{bmatrix},
\end{equation}
and $\mu_l \in [0,1)$ is the intra-block correlation coefficient of the $l$-th block. This equicorrelation model indicates that the ports belonging to the same block exhibit an identical correlation level, which aligns with the assumption that the spatial correlation can be regarded as approximately invariant within each block.

In this paper, the small-scale fading is considered to be Rayleigh fading. With the intra-block correlation matrix, the channels at the ports of the $l$-th block can be parameterized as
\begin{equation}\label{channel model-param}
h_{ln} =
\begin{cases}
\sigma x_{l0} + j\sigma y_{l0}, & n = 1,\\
\sigma\!\left(\sqrt{1-\mu_l^2}\,x_{ln} + \mu_l x_{l0}\right)\\
~+ j\sigma\!\left(\sqrt{1-\mu_l^2}\,y_{ln} + \mu_l y_{l0}\right), 
&n = 2,\ldots,N_l,
\end{cases}
\end{equation}
where $x_{ln}, y_{ln}, x_{l0}, y_{l0} \sim \mathcal{N}(0, \frac{1}{2})$ are mutually independent 
real Gaussian random variables, with $x_{l0}$ and $y_{l0}$ being the common components 
shared across all ports in block $l$.

Collecting the port channel coefficients of all blocks, the channel vector is given by
\begin{equation}\label{eq:channel_vector}
\mathbf{h} = [h_{11},\ldots,h_{1N_1},\, h_{21},\ldots,h_{2N_2},\,\ldots,\, h_{L1},\ldots,h_{LN_L}]^{T},
\end{equation}
whose spatial correlation matrix satisfies $\mathbb{E}\{\mathbf{h}\mathbf{h}^{H}\} = \sigma^2 \hat{\mathbf{\Sigma}}$, with $\hat{\mathbf{\Sigma}}$ defined in \eqref{block1}, and $\sigma^2$ denotes the average power of each channel coefficient. The block-diagonal structure of $\hat{\mathbf{\Sigma}}$ confirms that ports across different blocks are mutually independent, while ports within the same block are correlated according to the equicorrelation model in \eqref{block2}.

Based on the channel model, the received signal at the $n$-th port of the $l$-th block is given by
\begin{equation}\label{eq:received_signal_port}
y_{ln} = \sqrt{P}\, h_{ln}\, s + \eta_{ln}, 
\end{equation}
where $s$ denotes the transmitted information symbol with $\mathbb{E}\{|s|^2\} = 1$, $P$ is the transmit power, and $\eta_{ln} \sim \mathcal{CN}(0, \sigma_\eta^2)$ is the additive white Gaussian noise (AWGN) at the corresponding port. Thus, the received signal vector is given by
\begin{equation}\label{eq:received_signal_vector}
\mathbf{y} = \sqrt{P}\, \mathbf{h}\, s + \boldsymbol{\eta},
\end{equation}
where $\mathbf{y} \in \mathbb{C}^{N \times 1}$ and $\boldsymbol{\eta} \sim \mathcal{CN}(\mathbf{0}, \sigma_\eta^2 \mathbf{I}_N)$.

The instantaneous SNR at the $n$-th port of the $l$-th block $\gamma_{ln}$ is defined as
\begin{equation}\label{1}
\gamma_{ln} = \frac{P\, |h_{ln}|^2}{\sigma_\eta^2}.
\end{equation}
Since $h_{ln} \sim \mathcal{CN}(0,\sigma^2)$, the squared envelope $|h_{ln}|^2$ follows an exponential distribution with mean $\sigma^2$. Consequently, the instantaneous SNR $\gamma_{ln}$ is exponentially distributed with PDF
\begin{equation}
f_{\gamma_{ln}}(r) = \frac{1}{\bar{\gamma}}\, e^{-{r}/{\bar{\gamma}}}, \quad r \geq 0,
\end{equation}
where $\bar{\gamma} = P\sigma^2/\sigma_\eta^2$ denotes the average received SNR.

\section{Performance analysis}
In this section, we develop the performance analysis of the considered SIMO-FAS under the block-diagonal correlation model. 

\subsection{OP Analysis}

We begin with the OP analysis. 
The analysis starts with the derivation of the distribution associated with the selected-port SNR within a single block.
For the $l$-th block, the port experiencing the largest instantaneous SNR is selected, and the corresponding selected-port SNR is given by
\begin{equation}\label{max}
\gamma_l = \max\{\gamma_{l1},\ldots,\gamma_{lN_l}\}.
\end{equation}
Since ports across different blocks are mutually independent under the block-diagonal 
correlation model, the $L$ selected SNRs $\{\gamma_l\}_{l=1}^{L}$ are statistically decoupled, 
and the MRC output SNR is given by
\begin{equation}\label{eq:mrc_snr_sum}
\gamma_{\mathrm{MRC}} = \sum_{l=1}^{L} \gamma_l.
\end{equation}

Let $R$ denote the target transmission rate in bps/Hz, so that the corresponding 
SNR threshold is $\gamma_{\mathrm{th}} = 2^R - 1$. The outage event is defined as
\begin{equation}\label{eq:outage_event}
\mathcal{E} = \left\{\log_2\left(1+\gamma_{\mathrm{MRC}}\right) \le R\right\}
            = \left\{\gamma_{\mathrm{MRC}} \le \gamma_{\mathrm{th}}\right\},
\end{equation}
and the OP is accordingly given by
\begin{equation}\label{eq:outage_definition}
P_{\mathrm{out}}\left(\gamma_{\mathrm{th}}\right) 
= \Pr\!\left\{\gamma_{\mathrm{MRC}} \le \gamma_{\mathrm{th}}\right\} 
= F_{\gamma_{\mathrm{MRC}}}\left(\gamma_{\mathrm{th}}\right).
\end{equation}
Therefore, deriving the OP reduces to 
characterizing the CDF of each selected-port SNR $\gamma_l$ and subsequently obtaining the 
distribution of their sum.

\begin{lemma}\label{lemma:exact_cdf_single}
For the $l$-th block, the exact CDF of the selected-port SNR 
$\gamma_l$ is given by
\begin{equation}\label{eq:exact_cdf_single}
F_{\gamma_l}(x) = \int_{0}^{\frac{x}{\bar{\gamma}}} e^{-t} 
\prod_{n=2}^{N_l} \left[1 - Q_1\!\left(\alpha_l\sqrt{t},\, 
\beta_l\sqrt{\frac{x}{\bar{\gamma}}}\right)\right] dt,
\end{equation}
where $Q_1(\cdot,\cdot)$ is the first-order Marcum Q-function~\cite{marcumq}, 
$\alpha_l = \sqrt{\dfrac{2\mu_l^2}{1-\mu_l^2}}$, and 
$\beta_l = \sqrt{\dfrac{2}{1-\mu_l^2}}$.
\begin{proof}
Equation~\eqref{eq:exact_cdf_single} follows directly from \cite{fas} after 
the change of variable.
\end{proof}
\end{lemma}

The exact PDF of selected-port SNR can in principle be obtained by differentiating \eqref{eq:exact_cdf_single} with respect to $x$. However, since both the integration limit and the Marcum-$Q$ function arguments depend on $x$, the resulting expression is cumbersome and does not admit a useful closed-form. 
Thus, we retain the integral representation in~\eqref{eq:exact_cdf_single} and characterize the exact OP via either an $L$-fold convolution or a transform-domain approach.

\begin{theorem}\label{exop}
The exact OP of the SIMO-FAS system under the block-diagonal correlation model 
can be expressed as
\begin{equation}\label{eq:exact_mrc_convolution}
P_{\mathrm{out}}\left(\gamma_{\mathrm{th}}\right) 
= \int\!\cdots\!\int_{\sum_{l=1}^{L}\gamma_l \le \gamma_{\mathrm{th}}} 
\prod_{l=1}^{L} f_{\gamma_l}\left(\gamma_l\right)\,d\gamma_1\cdots d\gamma_L,
\end{equation}
which is an $L$-fold convolution integral that follows directly from the mutual 
independence of $\{\gamma_l\}_{l=1}^{L}$. Although exact, \eqref{eq:exact_mrc_convolution} becomes computationally prohibitive as the number of blocks $L$ grows. To avoid multidimensional integration, 
an equivalent representation in the frequency domain can be obtained by exploiting 
the characteristic function (CF). Specifically, let
\begin{equation}\label{eq:cf_single_block}
\Phi_{\gamma_l}\left(\omega\right) = \mathbb{E}\left[e^{j\omega\gamma_l}\right] 
= \int_{0}^{\infty} e^{j\omega x} f_{\gamma_l}\left(x\right)\,dx,
\end{equation}
denote the CF of selected-port SNR $\gamma_l$. By block independence, the CF of the MRC output is
\begin{equation}\label{eq:cf_mrc}
\Phi_{\mathrm{MRC}}\left(\omega\right) = \prod_{l=1}^{L} \Phi_{\gamma_l}\left(\omega\right),
\end{equation}
and invoking the Gil-Pelaez inversion theorem yields the exact OP as a 
single one-dimensional integral
\begin{equation}\label{eq:exact_mrc_gil_pelaez}
P_{\mathrm{out}}\left(\gamma_{\mathrm{th}}\right) 
= \frac{1}{2} - \frac{1}{\pi}\int_{0}^{\infty} \frac{1}{\omega}\,
\mathrm{Im}\!\left[e^{-j\omega\gamma_{\mathrm{th}}} 
\prod_{l=1}^{L}\Phi_{\gamma_l}\left(\omega\right)\right]d\omega.
\end{equation}
\end{theorem}

Although \eqref{eq:exact_mrc_convolution} and \eqref{eq:exact_mrc_gil_pelaez} 
characterize the exact OP, they remain analytically intractable, as they conceal the explicit dependence of system performance on key parameters. This motivates the derivation of closed-form high-SNR asymptotic expressions that expose the dominant performance scaling laws of the considered SIMO-FAS.

\begin{lemma}\label{Lemma1:A CDF}
For the high-SNR analysis, the CDF of the selected-port SNR $\gamma_l$ in the $l$-th block is approximated as
\begin{equation}\label{asym_cdf}
F_{\gamma_l}\left(\gamma\right) \approx 
C_l\left(\frac{\gamma}{\bar{\gamma}}\right)^{N_l},
\end{equation}
where $C_l \triangleq \prod_{n=2}^{N_l}\dfrac{1}{1-\mu_l^2}$.
\begin{proof}
See Appendix~A.
\end{proof}
\end{lemma}

\begin{lemma}\label{Lemma2:A PDF}
For the high-SNR analysis, the PDF of selected-port SNR $\gamma_l$ is approximated as
\begin{equation}\label{asym_pdf}
f_{\gamma_l}\left(\gamma\right) \approx 
\frac{C_l N_l}{\bar{\gamma}^{N_l}}\,\gamma^{N_l-1},
\end{equation}
which follows directly by differentiating \eqref{asym_cdf} with respect to 
$\gamma$.
\end{lemma}

\begin{lemma}\label{system_cdf}
For the high-SNR analysis, the CDF of the MRC output SNR $\gamma_{\mathrm{MRC}} = \sum_{l=1}^{L}\gamma_l$ is approximated as
\begin{equation}\label{eq:final_mrc_outage}
F_{\gamma_{\mathrm{MRC}}}\left(\gamma\right) \approx 
\frac{\prod_{l=1}^{L} N_l!}{\left(\sum_{l=1}^{L} N_l\right)!}
\left(\prod_{l=1}^{L} C_l\right)
\left(\frac{\gamma}{\bar{\gamma}}\right)^{\sum_{l=1}^{L} N_l}.
\end{equation}
\begin{proof}
See Appendix~B.
\end{proof}
\end{lemma}

\begin{corollary}\label{corollary:asymptotic_op}
The high-SNR asymptotic OP of the SIMO-FAS system is given by
\begin{equation}\label{eq:asymptotic_op}
P_{\mathrm{out}}^{\infty}\left(\gamma_{\mathrm{th}}\right) \approx
\frac{\prod_{l=1}^{L} N_l!}{\left(\sum_{l=1}^{L} N_l\right)!}
\left(\prod_{l=1}^{L} C_l\right)
\left(\frac{\gamma_{\mathrm{th}}}{\bar{\gamma}}\right)^{N}.
\end{equation}
\begin{proof}
The OP is found by substituting \(\gamma=\gamma_{\mathrm{th}}\)  into the joint CDF in~\eqref{eq:final_mrc_outage}, which completes the proof.
\end{proof}
\end{corollary}

\begin{theorem}\label{theorem:gamma_matching_op}
For the high-SNR analysis, the distribution of the combined SNR $\gamma_{\mathrm{MRC}}$ can be asymptotically matched to a Gamma distribution $\tilde{\gamma}\sim\mathrm{Gamma}(k,\theta)$, with shape parameter $k$ and scale parameter $\theta$ given by
\begin{equation}\label{eq:gamma_params}
k = N, \qquad
\theta = \frac{\bar{\gamma}}{\left(N!\,\Psi\right)^{\frac{1}{N}}},
\end{equation}
where $\Psi \triangleq \dfrac{\prod_{l=1}^{L}N_l!}{N!}\prod_{l=1}^{L}C_l$.

Consequently, the OP of the SIMO-FAS system admits the following closed-form approximation
\begin{equation}\label{eq:gamma_op_closed_form}
P_{\mathrm{out}}\left(\gamma_{\mathrm{th}}\right)
\approx \frac{\gamma\!\left(N, \frac{\gamma_{\mathrm{th}}}{\theta}\right)}{\Gamma(N)},
\end{equation}
where $\gamma(\cdot,\cdot)$ denotes the lower incomplete Gamma function. Since $N$ is a positive integer, \eqref{eq:gamma_op_closed_form} simplifies to
\begin{equation}\label{eq:gamma_op_closed_form_int}
P_{\mathrm{out}}\left(\gamma_{\mathrm{th}}\right)
\approx 1 - e^{-\frac{\gamma_{\mathrm{th}}}{\theta}}
\sum_{n=0}^{N-1} \frac{1}{n!}\left(\frac{\gamma_{\mathrm{th}}}{\theta}\right)^{\!n}.
\end{equation}
\begin{proof}
See Appendix~C.
\end{proof}
\end{theorem}

\begin{remark}
The closed-form expression \eqref{eq:gamma_op_closed_form_int} provides a more refined approximation of the OP compared to the purely asymptotic result in \textbf{Corollary}~\ref{corollary:asymptotic_op}. While \eqref{eq:asymptotic_op} captures only the leading-order term $\left(\gamma_{\mathrm{th}}/\bar{\gamma}\right)^{N}$ and is accurate only in the high-SNR regime, the Gamma-matching approximation in \eqref{eq:gamma_op_closed_form_int} retains the full Gamma CDF structure, which includes higher-order terms that yield improved agreement at moderate SNRs.
\end{remark}

\begin{proposition}\label{proposition1:do}
The diversity order of the SIMO-FAS system is
\begin{equation}\label{eq:do}
G_d = -\lim_{\bar{\gamma}\to\infty}
\frac{\log P_{\mathrm{out}}\left(\gamma_{\mathrm{th}}\right)}
{\log\bar{\gamma}} = N.
\end{equation}
\begin{proof}

By substituting $\Psi \triangleq \frac{\prod_{l=1}^{L}N_l!}{N!}\prod_{l=1}^{L}C_l$ and $N={\sum_{l=1}^{L} N_l}$ into \eqref{eq:asymptotic_op}, the high-SNR asymptotic OP simplifies to
\begin{equation}
P_{\mathrm{out}}\left(\gamma_{\mathrm{th}}\right)
\approx \Psi\left(\frac{\gamma_{\mathrm{th}}}{\bar{\gamma}}\right)^{\!N}.
\end{equation}
Taking the logarithm and substituting into \eqref{eq:do} gives
\begin{equation}
G_d \approx -\lim_{\bar{\gamma}\to\infty}\frac{\log\Psi + N\log\gamma_{\mathrm{th}} - N\log\bar{\gamma}}{\log\bar{\gamma}} = N.
\end{equation}
\end{proof}
\end{proposition}

\begin{remark}\label{remark2}
The diversity order, which is shown to approximate the number of ports $N$, is independent of the number of blocks $L$. This result lends itself to a natural physical interpretation. Increasing the number of ports directly enlarges the diversity order, which in turn steepens the high-SNR slope of the OP curve, while the block partitioning parameters, including the number of blocks, the block sizes, and the intra-block correlation, affect only the coding gain without altering the diversity order.
\end{remark}

\begin{proposition}\label{proposition2:equal_blocks}
Consider the equal-size-block assumption where $N_l = N/L$ for all
$l\in\{1,\ldots,L\}$ (with $N$ divisible by $L$) and $\mu_l = \mu$ for all $l$.
Specializing \textbf{Corollary}~\ref{corollary:asymptotic_op} to this case, the asymptotic OP reduces to
\begin{equation}\label{eq:outage_equal_blocks}
P_{\mathrm{out}}^{\infty}\left(\gamma_{\mathrm{th}}\right) \approx
\frac{\left(\left(N/L\right)!\right)^{L}}{N!}\,\left(1-\mu^2\right)^{-(N-L)}
\left(\frac{\gamma_{\mathrm{th}}}{\bar{\gamma}}\right)^{N}.
\end{equation}
\begin{proof}
Substituting $N_l = N/L$, $\mu_l = \mu$ into \textbf{Corollary}~\ref{corollary:asymptotic_op} yields the result.
\end{proof}
\end{proposition}

\begin{remark}
The expression in \eqref{eq:outage_equal_blocks}, which provides a compact closed-form approximation for the equal-size-block configuration, will be employed as the theoretical benchmark in the numerical results.
\end{remark}

\begin{proposition}\label{proposition:exact_L_eq_N}
When $L=N$, each block contains a single port, i.e., $N_l=1$ for all
$l\in\{1,\ldots,N\}$, so that no intra-block port selection is involved
and $\gamma_l = \gamma_{l1}$. The exact OP in this case is given by
\begin{equation}\label{eq:exact_op_L_eq_N}
P_{\mathrm{out}}\left(\gamma_{\mathrm{th}}\right)
= 1 - e^{-\frac{\gamma_{\mathrm{th}}}{\bar{\gamma}}}
\sum_{n=0}^{N-1}\frac{1}{n!}
\left(\frac{\gamma_{\mathrm{th}}}{\bar{\gamma}}\right)^{n}.
\end{equation}
\begin{proof}
When $L=N$, the exact CDF in \textbf{Lemma}~\ref{lemma:exact_cdf_single} reduces to
\begin{equation}
F_{\gamma_l}(x) = \int_{0}^{\frac{x}{\bar{\gamma}}} e^{-t}\,dt
= 1 - e^{-\frac{x}{\bar{\gamma}}},
\end{equation}
since the product term in \eqref{eq:exact_cdf_single} is empty when $N_l=1$.
Hence each selected-port SNR $\gamma_l$ is exponentially distributed with mean $\bar{\gamma}$.
Since ports across different blocks are mutually independent under the
block-diagonal correlation model, $\{\gamma_l\}_{l=1}^{N}$ are i.i.d.
exponential random variables, and their sum
$\gamma_{\mathrm{MRC}} = \sum_{l=1}^{N}\gamma_l$
follows a Gamma distribution with shape parameter $N$ and scale parameter
$\bar{\gamma}$. The CDF of this distribution directly yields
\begin{equation}
F_{\gamma_{\mathrm{MRC}}}\left(\gamma_{\mathrm{th}}\right)
= 1 - e^{-\frac{\gamma_{\mathrm{th}}}{\bar{\gamma}}}
\sum_{n=0}^{N-1}\frac{1}{n!}
\left(\frac{\gamma_{\mathrm{th}}}{\bar{\gamma}}\right)^{n},
\end{equation}
which, together with \eqref{eq:outage_definition}, completes the proof.
\end{proof}
\end{proposition}

\begin{remark}
The expression in \eqref{eq:exact_op_L_eq_N} coincides with the exact OP
of a conventional $N$-branch MRC receiver over i.i.d. Rayleigh fading
channels, whose diversity order is well known to be $N$ \cite{simon2005digital}. This is consistent
with the physical interpretation of the $L=N$ case. When every block contains
exactly one port, the block-diagonal correlation model degenerates to a fully independent setting, and the SIMO-FAS reduces to a standard $N$-branch MRC system with no intra-block correlation. Moreover, in this special case, the Gamma-matching result in \textbf{Theorem}~\ref{theorem:gamma_matching_op} yields $\Psi = 1/N!$ and $\theta = \bar{\gamma}$, so that \eqref{eq:gamma_op_closed_form_int} reduces exactly to \eqref{eq:exact_op_L_eq_N}, confirming the consistency of the Gamma-matching with the exact result in the fully independent setting. Since spatial correlation can only degrade performance relative to the fully independent case, the diversity order of the general correlated SIMO-FAS is upper-bounded by the number of ports $N$. \textbf{Proposition}~\ref{proposition1:do} confirms that this upper bound is in fact achieved for any block configuration, thereby establishing $G_d = N$ as the diversity order of the considered system.
\end{remark}

\subsection{ER Analysis}

To complement the preceding OP analysis, we now investigate the ER of the SIMO-FAS system. 

\begin{theorem}\label{ER_exact}
The exact ER of the SIMO-FAS system under the block-diagonal correlation model can be expressed as
\begin{equation}\label{er}
\begin{aligned}
R_{N} &= \mathbb{E}\left\{\log_2\left(1 + \gamma_{\mathrm{MRC}}\right)\right\} \\
      &= -\int_{0}^{\infty}\log_2\left(1+x\right)\,
         d\left(1 - F_{\gamma_{\mathrm{MRC}}}\left(x\right)\right) \\
      &= \frac{1}{\ln 2}\int_{0}^{\infty}
         \frac{1 - F_{\gamma_{\mathrm{MRC}}}\left(x\right)}{1+x}\,dx.
\end{aligned}
\end{equation}
\end{theorem}

\begin{remark}
Equation~\eqref{er} shows that the ER is fully determined by the CDF of the MRC output SNR $\gamma_{\mathrm{MRC}}$. Substituting the exact OP expression \eqref{eq:exact_mrc_gil_pelaez} into
\eqref{er} yields an exact integral representation of ER. However, since the CDF of the MRC output SNR $F_{\gamma_{\mathrm{MRC}}}$ is expressed as an integral involving a product of CFs, the resulting expression for ER does not admit a closed-form, thereby offering limited analytical insight.
\end{remark}

To circumvent this difficulty while still obtaining transparent and explicit results, we proceed
by substituting the Gamma-matching CDF from \textbf{Theorem}~\ref{theorem:gamma_matching_op}
into \eqref{er}, which yields a closed-form ER approximation.
\begin{theorem}\label{theorem:gamma_matching_er}
Based on the Gamma-matching in \textbf{Theorem}~\ref{theorem:gamma_matching_op}, the ER of
the SIMO-FAS system admits the following closed-form approximation
\begin{equation}\label{eq:gamma_matching_er}
R_{N}^{\infty} \approx \frac{e^{1/\theta}}{\ln 2}
\sum_{n=1}^{N} E_n\!\left(\frac{1}{\theta}\right),
\end{equation}
where $E_n(\cdot)$ denotes the generalized exponential integral of order $n$
and $\theta$ is given by \eqref{eq:gamma_params}.
\begin{proof}
See Appendix~D.
\end{proof}
\end{theorem}

\begin{remark}\label{remark5}
The Gamma-matching ER approximation in \eqref{eq:gamma_matching_er} captures
both the leading-order asymptotic behavior and the higher-order statistics of
the combined SNR. Nevertheless, since the Gamma-matching in
\textbf{Theorem}~\ref{theorem:gamma_matching_op} is obtained by matching the
CDF near the origin, it is well suited for the OP analysis where
$\gamma_{\mathrm{th}}$ is typically small, but the ER is an expectation over
the entire SNR range, which may lead to a larger approximation gap, especially
when $\bar{\gamma}$ is not sufficiently large.
\end{remark}

\begin{theorem}\label{theorem:ergodic_rate_tight}
The high-SNR asymptotic ER of the SIMO-FAS system with $L$ blocks is given by
\begin{equation}\label{eq:tight_ergodic_rate}
R_{N}^{\infty} \approx \log_2\!\left(1 + \bar{\gamma} L\left[\mu_l^2
+ \left(1-\mu_l^2\right)\sum_{n=1}^{N_l}\frac{1}{n}\right]\right).
\end{equation}
\begin{proof}
See Appendix~E.
\end{proof}
\end{theorem}

\begin{remark}\label{remark_er_close}
The approximation \eqref{eq:tight_ergodic_rate} reveals the capacity scaling laws of the SIMO-FAS system in a transparent manner. The effective mean SNR scales linearly with the number of blocks, reflecting the combining gain from MRC across independent blocks.
\end{remark}

\begin{remark} 
Within each block, the capacity contribution grows logarithmically through the harmonic series $\sum_{n=1}^{N_l}\frac{1}{n}$, but is discounted by the spatial independence factor $(1-\mu_l^2)$. Thus, stronger intra-block correlation suppresses the per-block diversity gain, which consequently diminishes the contribution of additional ports. These observations provide useful insights into the impact of block partitioning on the ER under different correlation conditions.
\end{remark}

\begin{proposition}\label{proposition:er_L_eq_N}
When $L = N$, the ER of the SIMO-FAS system admits the exact closed-form expression
\begin{equation}\label{eq:er_L_eq_N}
R_{N}\big|_{L=N} = \frac{e^{\frac{1}{\bar{\gamma}}}}{\ln 2}
\sum_{n=1}^{N} E_n\!\left(\frac{1}{\bar{\gamma}}\right),
\end{equation}
where $E_n(\cdot)$ denotes the generalized exponential integral of order $n$.
\begin{proof}
When $L = N$, \textbf{Proposition}~\ref{proposition:exact_L_eq_N} establishes that 
$\gamma_{\mathrm{MRC}}$ follows a Gamma distribution with shape $N$ and 
scale $\bar{\gamma}$, whose CDF is given by \eqref{eq:exact_op_L_eq_N}. 
Substituting into \eqref{er} yields
\begin{equation}
R_{N}\big|_{L=N} = \frac{1}{\ln 2}\int_{0}^{\infty}
\frac{1}{1+x}\,e^{-\frac{x}{\bar{\gamma}}}
\sum_{n=0}^{N-1}\frac{1}{n!}
\left(\frac{x}{\bar{\gamma}}\right)^{n}dx.
\end{equation}
Exchanging the order of summation and integration, the $n$-th term evaluates to
\begin{equation}
\frac{1}{n!\,\bar{\gamma}^n}\int_{0}^{\infty}
\frac{x^n\,e^{-\frac{x}{\bar{\gamma}}}}{1+x}\,dx
= e^{\frac{1}{\bar{\gamma}}}\,E_{n+1}\!\left(\frac{1}{\bar{\gamma}}\right),
\end{equation}
where the last equality follows from the standard integral representation of 
the generalized exponential integral \cite{integral}. Summing over 
$n = 0, 1, \ldots, N-1$ and re-indexing completes the proof.
\end{proof}
\end{proposition}

\begin{remark}\label{remark_er_ln}
The closed-form expression \eqref{eq:er_L_eq_N} coincides with the well-known exact ER of a conventional $N$-branch MRC receiver over i.i.d. Rayleigh fading channels \cite{Alouini1999}, which is consistent with the fact that the $L=N$ special case reduces to a fully independent $N$-branch MRC system as established in \textbf{Proposition}~\ref{proposition:exact_L_eq_N}. This result serves as both a sanity check for the general framework and a performance upper bound. Since spatial correlation degrades the ER relative to the independent case, the ER of the general correlated SIMO-FAS is bounded above by \eqref{eq:er_L_eq_N} for any block configuration with the same number of ports.
\end{remark}

\section{Numerical Results}\label{sec:simu}

Numerical investigations are conducted in this section to examine the accuracy of the developed analytical results and to explore the performance variations of the proposed SIMO-FAS under different system configurations. Unless otherwise specified, all results are obtained over Rayleigh fading channels under the block-diagonal correlation model with the equal-size-block assumption, and the intra-block correlation coefficients are taken to be identical, i.e., $\mu_l=\mu$ for all $l$. The Monte Carlo simulations in this paper use $1 \times 10^9$ channel realizations for OP and $1 \times 10^5$ channel realizations for ER.

\begin{figure}[t!]
\centering
\includegraphics[width=3.5in]{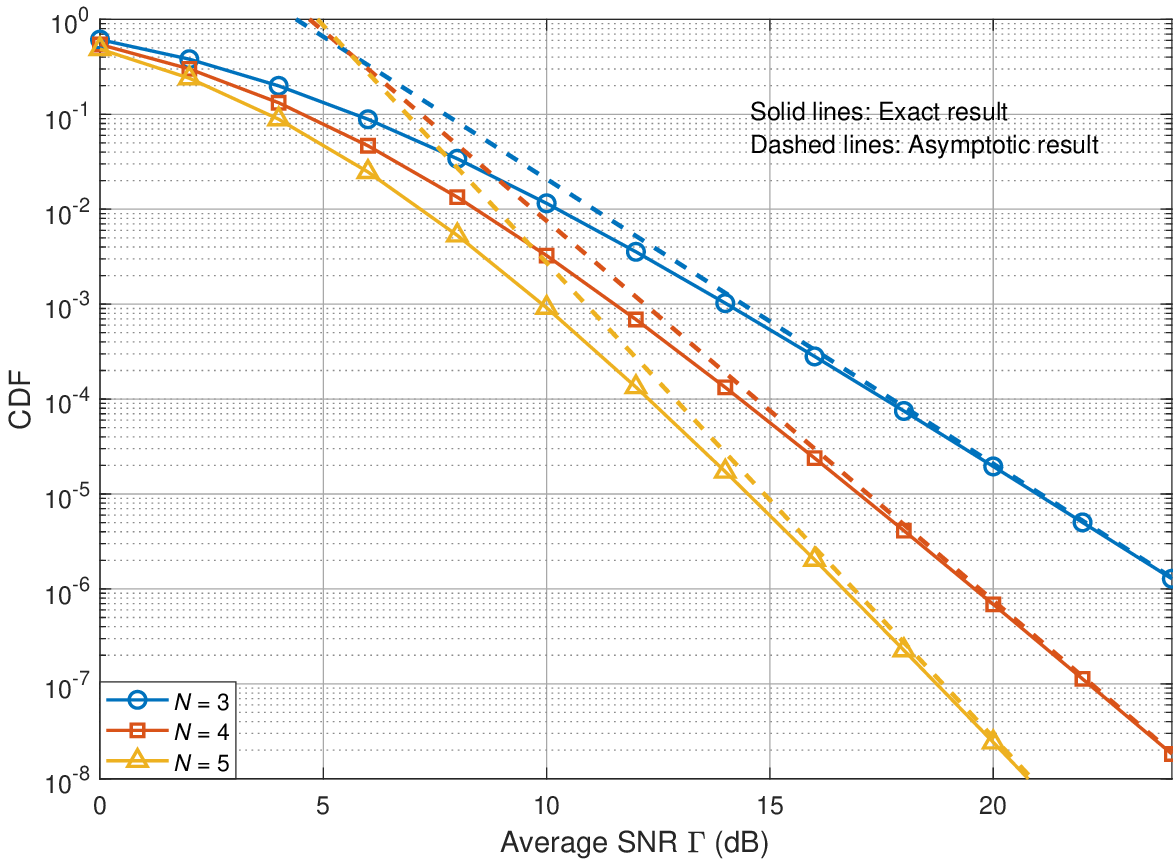}
\caption{The CDF of the selected-port SNR in the $l$-th block versus the average SNR with intra-block correlation coefficient $\mu=0.75$.}\label{cdf}
\vspace{-3mm}
\end{figure}

\subsection{Validation of the Asymptotic CDF of the Selected-Port SNR}
Fig.~\ref{cdf} shows the CDF of the selected-port SNR in the $l$-th block versus the average SNR for different numbers of ports, with intra-block correlation coefficient $\mu = 0.75$. The asymptotic results, plotted as dashed lines, agree closely with the exact analytical results in the high-SNR regime, validating Lemma~\ref{Lemma1:A CDF}. The small discrepancy in the low-SNR regime is expected, as the asymptotic expression retains only the dominant high-SNR term. As the number of ports increases, the CDF decays faster, since a larger number of ports raises the chance of selecting a strong port within each block. Fig.~\ref{cdf} therefore confirms both the accuracy of the asymptotic derivation and the diversity benefit of increasing the number of ports.

\begin{figure}[t!]
\centering
\includegraphics[width=3.5in]{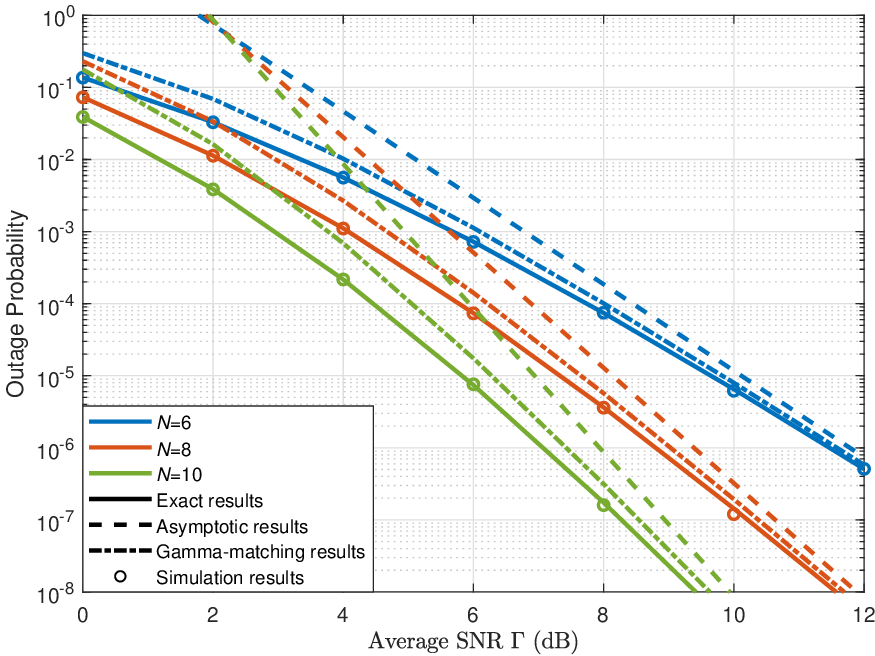}
\caption{The OP of SIMO-FAS versus the average SNR with the number of blocks $L=2$ with intra-block correlation coefficient $\mu=0.70$.}\label{SNR1}
\vspace{-3mm}
\end{figure}

\subsection{Impact of the Average SNR on OP}
Fig.~\ref{SNR1} shows the OP of the SIMO-FAS as a function of the average SNR for different numbers of ports, with $L=2$ blocks and intra-block correlation coefficient $\mu = 0.70$. The exact analytical results, shown as solid lines, match the Monte Carlo simulations depicted by circles across the entire SNR range, confirming the accuracy of the proposed framework. The Gamma-matching approximation, represented by dash-dot lines, is nearly indistinguishable from the exact curves over the moderate-to-high SNR range. By contrast, the asymptotic results plotted as dashed lines converge to the exact curves only as the average SNR grows large. This contrast reflects the fact that the asymptotic expression retains only the leading-order term, while the Gamma-matching approximation preserves the full Gamma CDF structure, which therefore remains accurate over a wider SNR range. The OP exhibits a monotonically decreasing trend with increasing average SNR, while a larger number of ports results in a more pronounced decline slope, which agrees with the higher diversity order derived in \textbf{Proposition}~\ref{proposition1:do}.

\begin{figure}[t!]
\centering
\includegraphics[width=3.5in]{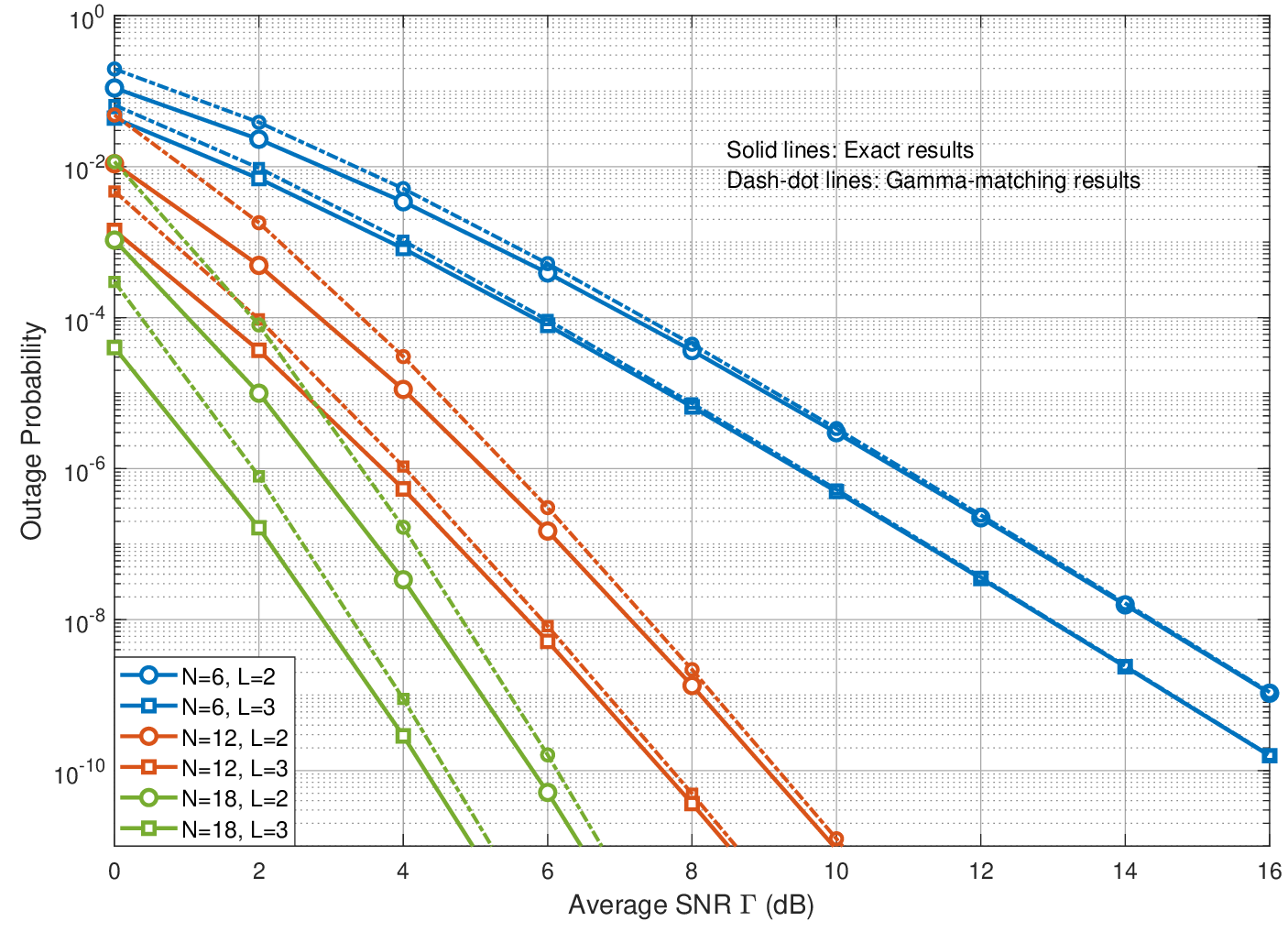}
\caption{OP of the SIMO-FAS versus the average SNR for different numbers of correlation blocks $L$ with intra-block correlation coefficient $\mu=0.60$.}\label{SNR2}
\vspace{-3mm}
\end{figure}

\subsection{Effect of the Number of Blocks on OP}
Fig.~\ref{SNR2} shows the OP of the SIMO-FAS versus the average SNR for different combinations of the number of ports and the number of blocks, with intra-block correlation coefficient $\mu = 0.60$. Both the exact analytical results and the Gamma-matching approximation agree closely with the simulations across the entire SNR range, confirming the accuracy of \textbf{Theorems}~\ref{exop} and~\ref{theorem:gamma_matching_op}. For a fixed number of ports, a larger number of blocks yields a lower OP, since partitioning the ports into more independent blocks creates more statistically independent branches for MRC combining. Notably, the curves for different numbers of blocks remain nearly parallel, which verifies \textbf{Remark}~\ref{remark2} that changing the number of blocks affects only the coding gain of the OP expression, while the diversity order is governed solely by the number of ports. Consequently, jointly increasing the number of ports and blocks yields a cumulative improvement. For instance, $N=18$ ports with $L=3$ blocks reduce the OP by several orders of magnitude relative to $N=6$ ports with $L=2$ blocks at the same SNR, reflecting the combined contribution of the two independent diversity mechanisms.

\begin{figure}[t!]
\centering
\includegraphics[width=3.5in]{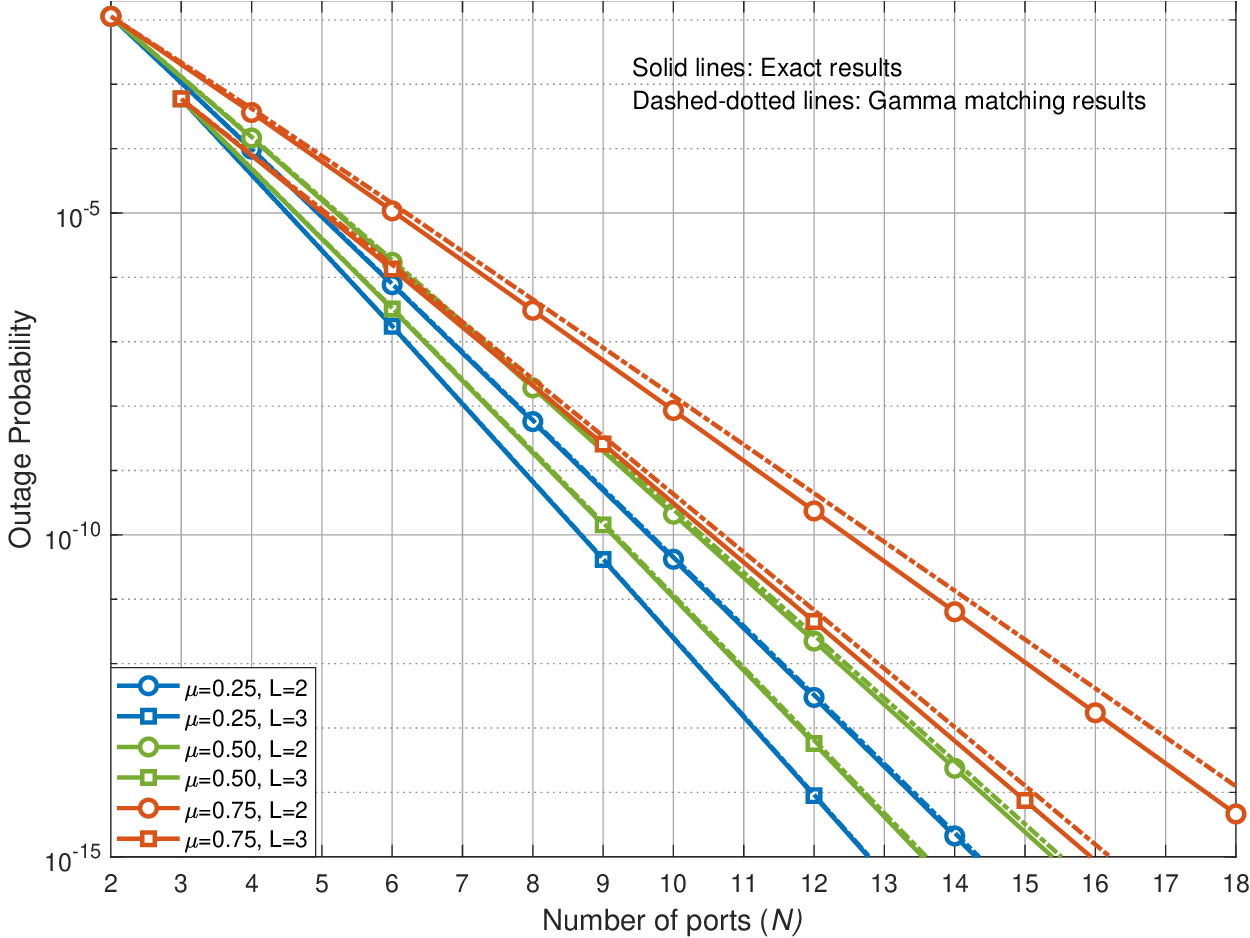}
\caption{OP of SIMO-FAS versus the number of ports under different intra-block correlation coefficients $\mu$ with $\bar{\gamma}=10$ dB.}\label{ns}
\vspace{-3mm}
\end{figure}

\subsection{Effect of the Number of Ports on OP}
The relationship between the OP and the number of ports is depicted in Fig.~\ref{ns}. 
The results are obtained under different intra-block correlation coefficients and correlation block numbers, while maintaining a constant average SNR of $\bar{\gamma}=10$ dB. The exact results show a close match with the Gamma-matching approximation across the entire range of the number of ports, validating the effectiveness of \textbf{Theorem}~\ref{theorem:gamma_matching_op}. As expected, the OP decreases continuously as the number of ports increases, since a larger port set provides richer spatial diversity. More importantly, stronger intra-block correlation degrades the OP for all configurations. The impact of the intra-block correlation coefficient is particularly pronounced at larger numbers of ports. For example, with $L=2$ blocks and $N=12$ ports, the OP under intra-block correlation coefficient $\mu=0.25$ is several orders of magnitude lower than under $\mu=0.75$. Meanwhile, for a fixed intra-block correlation coefficient, increasing the number of blocks from $2$ to $3$ yields a steeper decline of the OP with the number of ports, confirming that partitioning ports into more independent blocks enhances the system performance.

\begin{figure}[t!]
\centering
\includegraphics[width=3.5in]{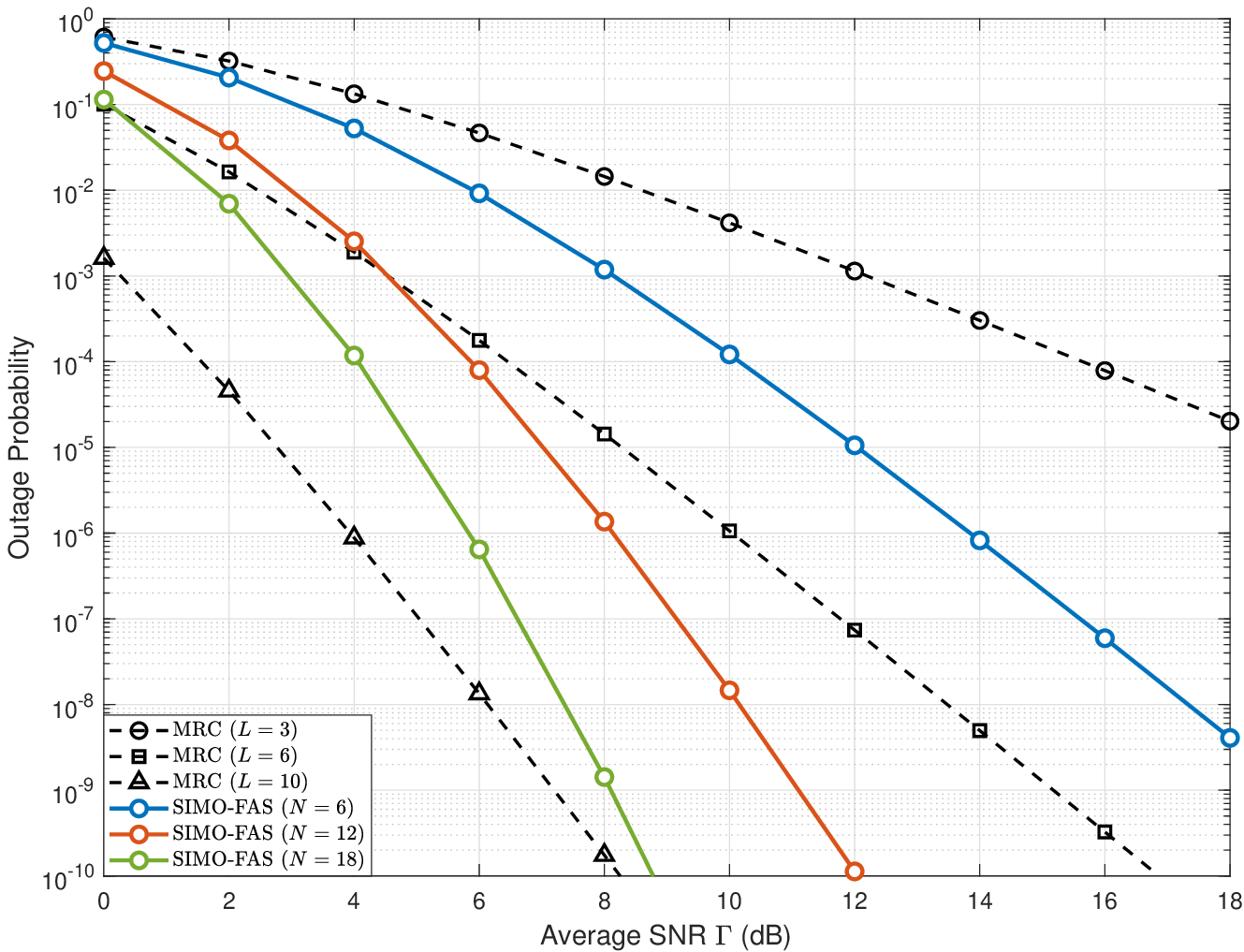}
\caption{The OP of SIMO-FAS system and conventional MRC receiver against the average SNR with intra-block correlation coefficient $\mu=0.75$.}\label{OPMRC}
\vspace{-3mm}
\end{figure}

\subsection{Comparison of the OP between SIMO-FAS and Conventional MRC Receiver}
In Fig.~\ref{OPMRC}, we compare the OP of the proposed SIMO-FAS receiver with that of the conventional MRC receiver as a function of the average SNR, with intra-block correlation coefficient $\mu=0.75$. The OP of the conventional MRC receiver corresponds to the special case $L = N$ established in \textbf{Proposition}~\ref{proposition:exact_L_eq_N}, since each block then contains a single port with no intra-block correlation. It can be observed that when the number of blocks is relatively small compared to the number of MRC branches, the conventional MRC receiver initially achieves a lower OP. However, as the number of ports increases, the per-block selection diversity strengthens and the two curves cross, with the SIMO-FAS eventually matching or outperforming conventional MRC despite having fewer combining branches. This crossover behavior highlights a key advantage of the FAS architecture. By deploying a larger number of ports within a compact aperture, the SIMO-FAS can compensate for the smaller number of combining branches, ultimately delivering competitive or superior reliability. Therefore, Fig.~\ref{OPMRC} demonstrates that SIMO-FAS can effectively trade port selection diversity for combining diversity, achieving comparable or better outage performance than conventional MRC with fewer radio-frequency chains.
 
\begin{figure}[t!]
\centering
\includegraphics[width=3.5in]{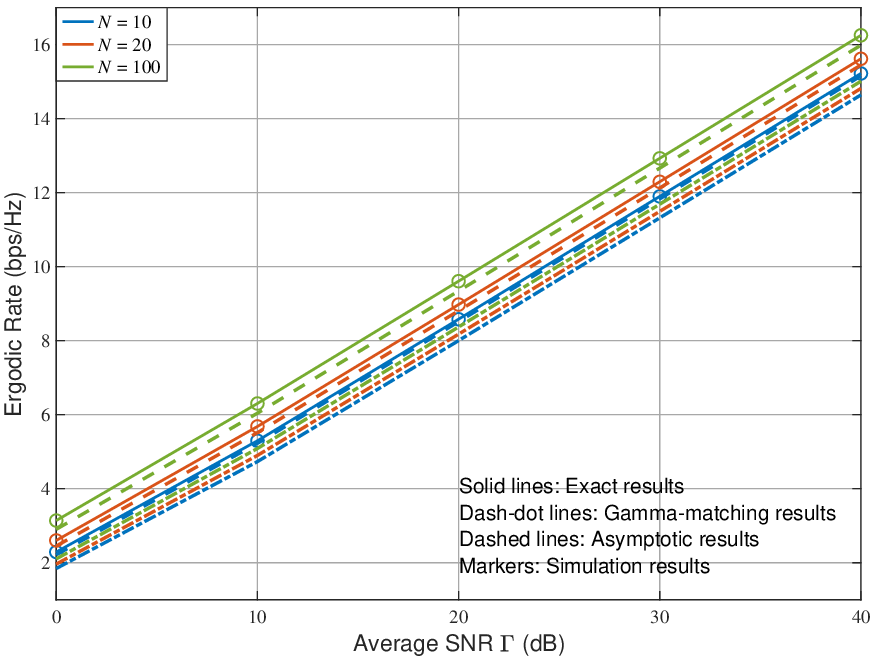}
\caption{ER of SIMO-FAS system versus the average SNR with intra-block correlation coefficient $\mu=0.60$.}\label{ergodic_rate}
\vspace{-3mm}
\end{figure}

\subsection{Effect of the Average SNR on ER}
In Fig.~\ref{ergodic_rate}, we present the ER versus the average SNR for the considered SIMO-FAS with intra-block correlation coefficient $\mu = 0.60$. It is observed that the Gamma-matching approximation represented by dash-dot lines exhibits a visible gap from the exact results across the entire SNR range. In contrast, the Jensen-based approximation depicted by dashed lines yields a much smaller gap, thereby providing a tighter estimate. This behavior is consistent with the discussion in \textbf{Remark}~\ref{remark5}. Since the Gamma-matching is obtained by matching the CDF near the origin, it is well suited for threshold-based metrics such as the OP, where the gap narrows as the average SNR increases. However, the ER is an expectation over the entire SNR distribution, so the local accuracy of the Gamma-matching near the origin does not translate into a uniformly tight ER approximation, and the gap persists even at high SNR. On the other hand, since the Jensen-based approximation directly captures the mean SNR, it offers a more reliable approximation for the ER. As expected, the ER increases monotonically with the average SNR because a higher received signal quality directly improves the achievable transmission rate. In addition, the configurations with more blocks and ports consistently achieve higher ER values, indicating that the gains brought by port selection and MRC are not limited to reliability improvement, but also translate into a meaningful enhancement in average spectral efficiency. Nevertheless, due to the intrinsic limitation of the Gamma-matching approximation discussed above, the following figures adopt the Jensen-based approximation as the primary analytical tool for rate characterization.

\begin{figure}[t!]
\centering
\includegraphics[width=3.5in]{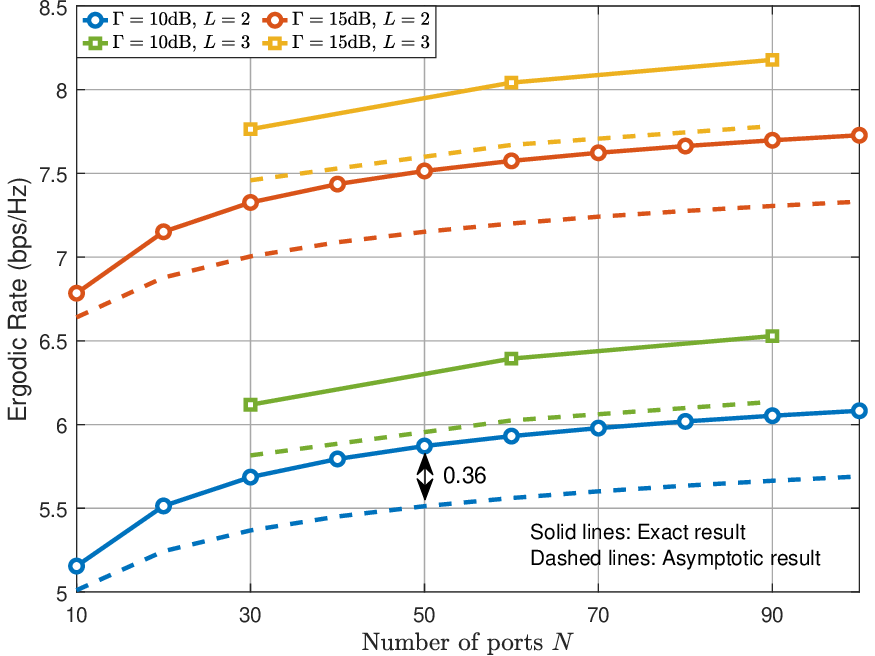}
\caption{ER of the proposed SIMO-FAS system versus the number of ports with intra-block correlation coefficient $\mu=0.75$.}\label{ergodic_rate_ports}
\vspace{-3mm}
\end{figure}

\subsection{Effect of the Number of Ports on ER}
The impact of the number of ports on the ER is illustrated in Fig.~\ref{ergodic_rate_ports}. 
The evaluation is conducted under an intra-block correlation coefficient of $\mu = 0.75$ and a constant average SNR of $\bar{\gamma}=15$ dB. The Jensen-based approximation exhibits a consistently small gap from the exact results across the whole range of the number of ports. For example, at $N = 6$ ports, the exact ER is approximately $3.0$ bps/Hz while the Jensen-based approximation yields around $3.05$ bps/Hz, corresponding to a relative gap of less than $2\%$. As the number of ports increases, the gap remains within this narrow margin, confirming that the Jensen-based approximation provides a reliable and accurate characterization of the ER. It is observed that the ER increases with the number of ports, since a larger number of ports enlarges the selection space in each block, which in turn improves the effective SNR after MRC. Furthermore, the rate improvement becomes increasingly noticeable when the number of ports increases, since the enhanced selection diversity can be more efficiently exploited to achieve spectral-efficiency gains. This observation is also in agreement with the outage performance illustrated in Fig.~\ref{SNR2} and Fig.~\ref{ns}.

\begin{figure}[t!]
\centering
\includegraphics[width=3.5in]{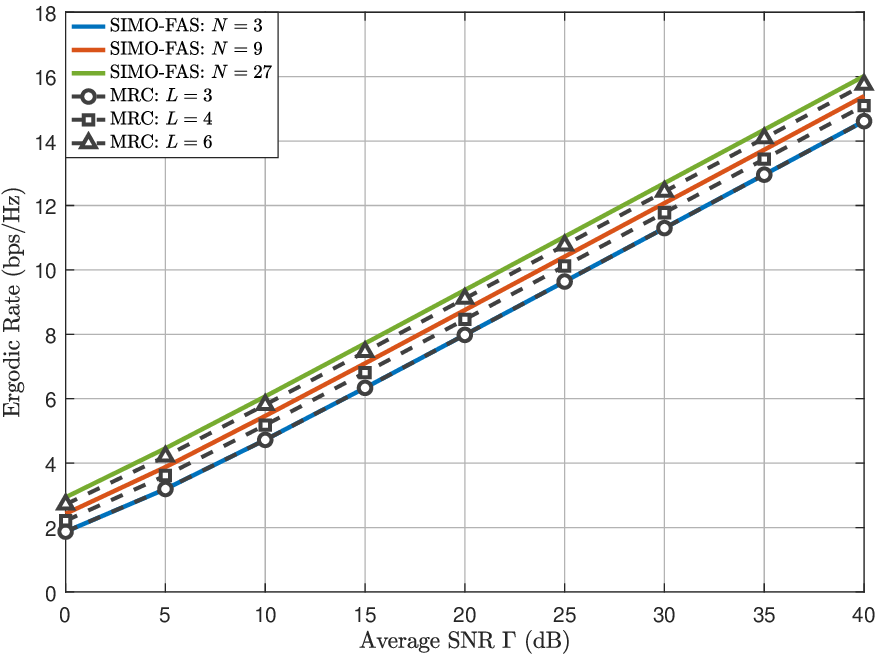}
\caption{The ER of SIMO-FAS system and conventional MRC receiver against the average SNR with intra-block correlation coefficient $\mu=0.75$.}\label{ermrc}
\vspace{-3mm}
\end{figure}

\subsection{Comparison of the ER between SIMO-FAS and Conventional MRC Receiver}
Fig.~\ref{ermrc} compares the ER of the proposed SIMO-FAS system with that of the conventional MRC receiver as a function of the average SNR, with intra-block correlation coefficient $\mu = 0.75$. The ER of the conventional MRC receiver is obtained from the exact closed-form expression in \textbf{Proposition}~\ref{proposition:er_L_eq_N} by setting $L = N$, which corresponds to the well-known ER of $N$-branch MRC over i.i.d. Rayleigh fading channels. It can be observed that the ER of both schemes increases monotonically with the average SNR, and that a larger number of ports consistently yields a higher ER, confirming that deploying more ports is beneficial for improving the average spectral efficiency. The analytical curves based on the Jensen-based approximation follow the exact results and simulation markers closely, validating the accuracy of the approximation. However, for the same number of combining branches, the proposed SIMO-FAS achieves a strictly higher ER than the conventional MRC receiver across the entire SNR range considered. This performance advantage originates from the same mechanism identified in the OP analysis: prior to MRC combining, the strongest port is selected within each block, so the branches entering the combiner carry statistically higher SNRs than those in the conventional MRC receiver, resulting in an improved combined output SNR and hence a higher ER. Furthermore, the ER gap between the two schemes remains approximately constant in the high-SNR regime, which is consistent with the fact that both systems achieve the same multiplexing gain of one, so the performance difference manifests as a fixed coding gain offset rather than a difference in slope. Therefore, Fig.~\ref{ermrc} verifies that the proposed SIMO-FAS not only inherits the combining gain of conventional MRC, but also achieves a consistent ER improvement through port selection.

\section{Conclusion}\label{sec:conclusion}
This paper has developed an analytical framework for characterizing the OP and ER of the SIMO-FAS with multiple activated ports under the block-diagonal correlation model. Exact OP expressions were derived in both convolution and characteristic-function forms. A Gamma-matching approximation was further developed to provide a refined closed-form OP estimate across the entire SNR range. High-SNR asymptotic analysis established that the diversity order equals the number of ports. For the ER, a Gamma-matching approximation together with a complementary Jensen-based approximation was derived, with the latter shown to provide tight estimates. Exact closed-form ER expressions were further obtained for the special case where the number of ports equals the number of blocks. Numerical evaluations demonstrate that the proposed SIMO-FAS system can attain performance comparable to that of conventional MRC receivers while requiring fewer combining branches. Moreover, the performance gain becomes increasingly significant with a growing number of ports, demonstrating the effectiveness of exploiting port selection diversity. These results establish SIMO-FAS as a cost-efficient alternative to conventional MRC for next-generation multiple access. Several promising directions remain open for further investigation, including the extension of the current framework to non-stationary channel models, the optimization of port allocation strategies, and the exploration of multi-user FAS scenarios.

\numberwithin{equation}{section}

\numberwithin{equation}{section}
\section*{Appendix~A: Proof of \textbf{Lemma}~\ref{Lemma1:A CDF}}\label{Appendix:Bs}
\renewcommand{\theequation}{A.\arabic{equation}}
\setcounter{equation}{0}

In the high-SNR regime, let $U \triangleq \gamma/\bar{\gamma}$. As $\bar{\gamma} \to \infty$, we have $U \to 0$, which implies $t \in [0,U]$ and $t \to 0$ as well. Hence $e^{-t} \approx 1$.

For the first-order Marcum Q-function $Q_1(a,b)$, the Bessel-series representation yields the dominant small-argument behavior ($a \to 0, b \to 0$) \cite{simon2005digital}
\begin{equation}\label{app}
    1 - Q_1(a,b) \approx \frac{b^2}{2}.
\end{equation}

Substituting $b = \beta_l \sqrt{U}$ into \eqref{app}, we obtain
\begin{equation}\label{A2}
    1 - Q_1\left( \alpha_l \sqrt{t}, \beta_l \sqrt{U} \right)
    \approx \frac{\beta_l^2 U}{2}
    = \frac{U}{1 - \mu_l^2}.
\end{equation}
The expression in \eqref{A2} is independent of the integration variable $t$. Therefore, the product term becomes a constant and
\begin{equation}\label{eq:asym_cdf_single}
\begin{aligned}
    F_{\gamma_l}^{\infty}(\gamma)
    &\approx \int_{0}^{U} \prod_{n=2}^{N_l} \left( \frac{U}{1 - \mu_l^2} \right) dt\\
    &= \left( \prod_{n=2}^{N_l} \frac{1}{1 - \mu_l^2} \right) U^{N_l}\\
    &= C_l \left( \frac{\gamma}{\bar{\gamma}} \right)^{N_l}.\\
\end{aligned}
\end{equation}
Thus, the proof is complete.

\numberwithin{equation}{section}
\section*{Appendix~B: Proof of \textbf{Lemma}~4}\label{Appendix:Cs}
\renewcommand{\theequation}{B.\arabic{equation}}
\setcounter{equation}{0}

Finding the exact distribution of $\gamma_{\mathrm{MRC}}$ requires an $L$-fold convolution of the PDFs, which is mathematically intractable in its exact form. However, in the asymptotic high-SNR regime, we can invoke the asymptotic convolution theorem for multi-block diversity systems \cite{wang2003simple}. If a set of independent non-negative random variables $\{X_l\}_{l=1}^L$ possesses PDFs that behave asymptotically near the origin as $f_{X_l}(x) \approx a_l x^{t_l}$ (with $a_l > 0$ and $t_l > -1$), the CDF of their sum $Y = \sum_{l=1}^L X_l$ at a threshold $y_{\mathrm{th}}$ is asymptotically given by:
\begin{equation}\label{eq:convolution_theorem_general}
    F_Y^{\infty}(y_{\mathrm{th}}) \approx \frac{\prod_{l=1}^{L} \left( a_l  t_l! \right)}{\left( \sum_{l=1}^{L} (t_l + 1) \right)!} y_{\mathrm{th}}^{\sum_{l=1}^{L} (t_l + 1)}.
\end{equation}

We now map the parameters from the PDF of the selected-port SNR \eqref{asym_pdf} into \eqref{eq:convolution_theorem_general}. For the $l$-th block, we identify the coefficient and the exponent as:
\begin{equation}
\begin{aligned}
    a_l &= \frac{C_l N_l}{\bar{\gamma}^{N_l}}, \\
    t_l &= N_l - 1.
\end{aligned}
\end{equation}

Based on this parameter mapping, we evaluate the component terms required for the general theorem. First, the term inside the product in the numerator becomes:
\begin{equation}\label{eq:numerator_eval}
    a_l  t_l! = \left( \frac{C_l  N_l}{\bar{\gamma}^{N_l}} \right) (N_l - 1)! = \frac{C_l  N_l!}{\bar{\gamma}^{N_l}}.
\end{equation}
Second, the summation term in the exponent and the denominator simplifies to
\begin{equation}\label{eq:denominator_eval}
    \sum_{l=1}^{L} (t_l + 1) = \sum_{l=1}^{L} \left( (N_l - 1) + 1 \right) = \sum_{l=1}^{L} N_l.
\end{equation}

Substituting the evaluated components \eqref{eq:numerator_eval} and \eqref{eq:denominator_eval} into the asymptotic convolution formula \eqref{eq:convolution_theorem_general}, we obtain the asymptotic OP of the MRC system evaluated at the threshold $\gamma_{\mathrm{th}}$:
\begin{equation}
\begin{aligned}
    P_{\mathrm{out, MRC}}^{\infty}(\gamma_{\mathrm{th}}) &\approx \frac{\prod_{l=1}^{L} \left( \frac{C_l  N_l!}{\bar{\gamma}^{N_l}} \right)}{\left( \sum_{l=1}^{L} N_l \right)!} \gamma_{\mathrm{th}}^{\sum_{l=1}^{L} N_l} \\
    &= \frac{\prod_{l=1}^{L} N_l!}{\left( \sum_{l=1}^{L} N_l \right)!} \left( \prod_{l=1}^{L} C_l \right) \left( \frac{\gamma_{\mathrm{th}}}{\bar{\gamma}} \right)^{\sum_{l=1}^{L} N_l}.
\end{aligned}
\end{equation}

\section*{Appendix~C: Proof of \textbf{Theorem}~\ref{theorem:gamma_matching_op}}\label{Appendix:gamma}
\renewcommand{\theequation}{C.\arabic{equation}}
\setcounter{equation}{0}

We prove the Gamma-matching result by equating the small-argument asymptotic expansion of the CDF of $\gamma_{\mathrm{MRC}}$ to that of a Gamma-distributed random variable.

From \textbf{Lemma}~\ref{system_cdf}, the asymptotic CDF of $\gamma_{\mathrm{MRC}}$ near the origin is given by
\begin{equation}\label{eq:gamma_matching_asym_cdf}
F_{\gamma_{\mathrm{MRC}}}(x)
\approx \frac{\prod_{l=1}^{L} N_l!}{N!}
\left(\prod_{l=1}^{L} C_l\right)
\left(\frac{x}{\bar{\gamma}}\right)^{N},
\end{equation}
where $N = \sum_{l=1}^{L} N_l$.

Now consider a Gamma-distributed random variable $\tilde{\gamma}\sim\mathrm{Gamma}(k,\theta)$ with shape parameter $k>0$ and scale parameter $\theta>0$, whose PDF and CDF are respectively
\begin{equation}
f_{\tilde{\gamma}}(x) = \frac{x^{k-1}e^{-x/\theta}}{\theta^{k}\,\Gamma(k)},\qquad
F_{\tilde{\gamma}}(x) = \frac{\gamma(k, x/\theta)}{\Gamma(k)},\quad x\ge 0,
\end{equation}
where $\gamma(\cdot,\cdot)$ denotes the lower incomplete Gamma function.

The lower incomplete Gamma function admits the series expansion \cite{abramowitz2006handbook}
\begin{equation}\label{eq:lower_inc_gamma_series}
\gamma(k, z) = \sum_{n=0}^{\infty} \frac{(-1)^n z^{k+n}}{n!\,(k+n)},
\qquad |z|<\infty.
\end{equation}
For small $z$, the leading-order behavior is
\begin{equation}\label{eq:lower_inc_gamma_leading}
\gamma(k, z) = \frac{z^{k}}{k} + O\!\left(z^{k+1}\right),\qquad z\to 0.
\end{equation}
Therefore, as $x\to 0$,
\begin{equation}\label{eq:gamma_cdf_asymptotic}
F_{\tilde{\gamma}}(x)
= \frac{\gamma(k, x/\theta)}{\Gamma(k)}
\approx \frac{(x/\theta)^{k}}{k\,\Gamma(k)}
= \frac{x^{k}}{\theta^{k}\,\Gamma(k+1)}.
\end{equation}

For the distribution of $\tilde{\gamma}$ to asymptotically match that of $\gamma_{\mathrm{MRC}}$ in the high-SNR regime, we require that the leading-order terms of their CDFs coincide. Equating the exponents of $x$ in \eqref{eq:gamma_matching_asym_cdf} and \eqref{eq:gamma_cdf_asymptotic} yields
\begin{equation}\label{eq:gamma_match_shape}
k = N.
\end{equation}
Equating the coefficients gives
\begin{equation}
\frac{1}{\theta^{N}\,\Gamma(N+1)}
= \frac{\prod_{l=1}^{L} N_l!}{N!}
\left(\prod_{l=1}^{L} C_l\right)
\frac{1}{\bar{\gamma}^{N}}.
\end{equation}
Since $\Gamma(N+1)=N!$, this simplifies to
\begin{equation}
\frac{1}{\theta^{N}}
= \left(\prod_{l=1}^{L} N_l!\right)
\left(\prod_{l=1}^{L} C_l\right)
\frac{1}{\bar{\gamma}^{N}}.
\end{equation}
Solving for $\theta$, we obtain
\begin{equation}
\theta^{N} = \frac{\bar{\gamma}^{N}}
{\prod_{l=1}^{L} N_l! \prod_{l=1}^{L} C_l},
\qquad
\theta = \frac{\bar{\gamma}}
{\left(\prod_{l=1}^{L} N_l! \prod_{l=1}^{L} C_l\right)^{1/N}}.
\end{equation}

Defining $\Psi \triangleq \dfrac{\prod_{l=1}^{L}N_l!}{N!}\prod_{l=1}^{L}C_l$, we have $N!\Psi = \prod_{l=1}^{L} N_l! \prod_{l=1}^{L} C_l$, and $\theta$ can be compactly written as
\begin{equation}
\theta = \frac{\bar{\gamma}}{\left(N!\,\Psi\right)^{1/N}}.
\end{equation}

Thus $\tilde{\gamma}\sim\mathrm{Gamma}(N,\theta)$ asymptotically matches $\gamma_{\mathrm{MRC}}$ in the high-SNR regime, which completes the proof of \textbf{Theorem}~\ref{theorem:gamma_matching_op}. The closed-form OP expressions in \textbf{Theorem}~\ref{theorem:gamma_matching_op} follow directly from the CDF of the Gamma distribution and the identity
\begin{equation}\label{eq:gamma_int_cdf}
\frac{\gamma(N, x)}{\Gamma(N)}
= 1 - e^{-x}\sum_{n=0}^{N-1}\frac{x^{n}}{n!},
\qquad N\in\mathbb{Z}^{+},
\end{equation}
by substituting $x = \gamma_{\mathrm{th}}/\theta$.

\numberwithin{equation}{section}
\section*{Appendix~D: Proof of \textbf{Theorem}~\ref{theorem:gamma_matching_er}}\label{Appendix:gamma_er}
\renewcommand{\theequation}{D.\arabic{equation}}
\setcounter{equation}{0}

From \textbf{Theorem}~\ref{theorem:gamma_matching_op}, the CDF of $\gamma_{\mathrm{MRC}}$ is approximated by that of a Gamma$(N,\theta)$ random variable:
\begin{equation}\label{eq:gamma_er_cdf}
F_{\gamma_{\mathrm{MRC}}}(x) \approx 1 - e^{-x/\theta}
\sum_{m=0}^{N-1} \frac{1}{m!}\left(\frac{x}{\theta}\right)^{\!m},
\end{equation}
where $\theta$ is given by \eqref{eq:gamma_params}. Substituting into the ER definition \eqref{er} yields
\begin{equation}\label{eq:gamma_er_int}
R_{N}^{\infty} = \frac{1}{\ln 2}\int_{0}^{\infty}
\frac{1}{1+x}\,e^{-x/\theta}
\sum_{m=0}^{N-1}\frac{1}{m!}
\left(\frac{x}{\theta}\right)^{\!m}dx.
\end{equation}
Exchanging the order of summation and integration, we obtain
\begin{equation}\label{eq:gamma_er_sum}
R_{N}^{\infty} = \frac{1}{\ln 2}\sum_{m=0}^{N-1}
\frac{1}{m!\,\theta^{m}}
\int_{0}^{\infty}\frac{x^{m}\,e^{-x/\theta}}{1+x}\,dx.
\end{equation}
The integral admits the closed-form evaluation \cite{integral}
\begin{equation}\label{eq:gamma_er_int_eval}
\frac{1}{m!\,\theta^{m}}
\int_{0}^{\infty}\frac{x^{m}\,e^{-x/\theta}}{1+x}\,dx
= e^{1/\theta}\,E_{m+1}\!\left(\frac{1}{\theta}\right),
\end{equation}
where $E_{n}(\cdot)$ denotes the generalized exponential integral. Substituting \eqref{eq:gamma_er_int_eval} into \eqref{eq:gamma_er_sum} and re-indexing $n=m+1$ gives
\begin{equation}
R_{N}^{\infty} = \frac{e^{1/\theta}}{\ln 2}
\sum_{m=0}^{N-1} E_{m+1}\!\left(\frac{1}{\theta}\right)
= \frac{e^{1/\theta}}{\ln 2}
\sum_{n=1}^{N} E_{n}\!\left(\frac{1}{\theta}\right),
\end{equation}
which completes the proof of \textbf{Theorem}~\ref{theorem:gamma_matching_er}.

\numberwithin{equation}{section}
\section*{Appendix~E: Proof of \textbf{Theorem}~\ref{theorem:ergodic_rate_tight}}\label{Appendix:Gs}
\renewcommand{\theequation}{E.\arabic{equation}}
\setcounter{equation}{0}

Since $f(x)=\log_2(1+x)$ is concave for $x\ge 0$, Jensen's inequality gives
\begin{equation}\label{eq:jensen_app_F}
    R_N \le \log_2 \left( 1 + \mathbb{E}\{\gamma_{\mathrm{MRC}}\} \right).
\end{equation}

Using the linearity of expectation and the equal-size-block assumption, we have
\begin{equation}\label{eq:exp_mrc}
    \mathbb{E}\{\gamma_{\mathrm{MRC}}\} = \mathbb{E} \left\{ \sum_{l=1}^{L} \gamma_l \right\} = L\mathbb{E}\{\gamma_l\}.
\end{equation}
Based on \eqref{max},  $\gamma_l = \mathop {\max }\limits_{1 \le k \le {N_l}}  \gamma_{l,k}$ denotes the selected SNR within the $l$-th block.

Approximating the correlated selection gain by a convex combination of the fully correlated and fully independent cases yields
\begin{equation}\label{eq:snr_decomp}
    \mathbb{E} \left\{ \max_{1 \le k \le N_l} \gamma_{l,k} \right\} \approx \mu_l^2 \, \mathbb{E}\{\gamma_c\} + (1 - \mu_l^2) \, \mathbb{E} \left\{ \max_{1 \le k \le N_l} \gamma_{\mathrm{ind},k} \right\},
\end{equation}
where $\gamma_c$ denotes the fully correlated component, and $\gamma_{\mathrm{ind},k}$ denotes the fully independent component. Here, $\mathbb{E}\{\gamma_c\}=\bar{\gamma}$, and by the order-statistics result for the maximum of $N_l$ i.i.d. exponential random variables \cite[Eq. (9.32)]{simon2005digital},
\begin{equation}\label{eq:harmonic_sum}
    \mathbb{E} \left\{ \max_{1 \le k \le N_l} \gamma_{\mathrm{ind},k} \right\} = \bar{\gamma} \sum_{n=1}^{N_l} \frac{1}{n}.
\end{equation}

Substituting these into \eqref{eq:snr_decomp}, we obtain
\begin{equation}\label{eq:exp_gamma_l}
    \mathbb{E}\{\gamma_l\} \approx \bar{\gamma} \left[ \mu_l^2 + (1 - \mu_l^2) \sum_{n=1}^{N_l} \frac{1}{n} \right].
\end{equation}

Finally, combining \eqref{eq:jensen_app_F}, \eqref{eq:exp_mrc}, and \eqref{eq:exp_gamma_l} gives
\begin{equation}
    R_{N}^{\infty} \approx \log_2 \left( 1 + \bar{\gamma} L \left[ \mu_l^2 + (1 - \mu_l^2) \sum_{n=1}^{N_l} \frac{1}{n} \right] \right).
\end{equation}
This completes the proof.

\bibliographystyle{IEEEtran}

\begin{thebibliography}{}
\providecommand{\url}[1]{#1}
\csname url@samestyle\endcsname
\providecommand{\newblock}{\relax}
\providecommand{\bibinfo}[2]{#2}
\providecommand{\BIBentrySTDinterwordspacing}{\spaceskip=0pt\relax}
\providecommand{\BIBentryALTinterwordstretchfactor}{4}
\providecommand{\BIBentryALTinterwordspacing}{\spaceskip=\fontdimen2\font plus
\BIBentryALTinterwordstretchfactor\fontdimen3\font minus
  \fontdimen4\font\relax}
\providecommand{\BIBforeignlanguage}[2]{{%
\expandafter\ifx\csname l@#1\endcsname\relax
\typeout{** WARNING: IEEEtran.bst: No hyphenation pattern has been}%
\typeout{** loaded for the language `#1'. Using the pattern for}%
\typeout{** the default language instead.}%
\else
\language=\csname l@#1\endcsname
\fi
#2}}
\providecommand{\BIBdecl}{\relax}
\BIBdecl

\end{thebibliography}


\begin{thebibliography}{10}
\providecommand{\url}[1]{#1}
\csname url@samestyle\endcsname
\providecommand{\newblock}{\relax}
\providecommand{\bibinfo}[2]{#2}
\providecommand{\BIBentrySTDinterwordspacing}{\spaceskip=0pt\relax}
\providecommand{\BIBentryALTinterwordstretchfactor}{4}
\providecommand{\BIBentryALTinterwordspacing}{\spaceskip=\fontdimen2\font plus
\BIBentryALTinterwordstretchfactor\fontdimen3\font minus
  \fontdimen4\font\relax}
\providecommand{\BIBforeignlanguage}[2]{{%
\expandafter\ifx\csname l@#1\endcsname\relax
\typeout{** WARNING: IEEEtran.bst: No hyphenation pattern has been}%
\typeout{** loaded for the language `#1'. Using the pattern for}%
\typeout{** the default language instead.}%
\else
\language=\csname l@#1\endcsname
\fi
#2}}
\providecommand{\BIBdecl}{\relax}
\BIBdecl


\bibitem{6G1}
C.-X.~Wang {\em et al.}, ``On the road to {6G}: Visions, requirements, key technologies and testbeds,'' \emph{IEEE Commun. Surveys Tuts.}, vol.~25, no.~2, pp. 905--974, 2nd Quart., 2023.

\bibitem{6G2}
W.~Jiang, B.~Han, M.~A.~Habibi, and H.~D.~Schotten, ``The road towards {6G}: A comprehensive survey,'' \emph{IEEE Open J. Commun. Soc.}, vol.~2, pp. 334--366, 2021.

\bibitem{6G3}
M.~Z.~Chowdhury, M.~Shahjalal, S.~Ahmed, and Y.~M.~Jang, ``{6G} wireless communication systems: Applications, requirements, technologies, challenges, and research directions,'' \emph{IEEE Open J. Commun. Soc.}, vol.~1, pp. 957--975, 2020.

\bibitem{6G4}
Y.~Xiao, Q.~Du, and G.~K.~Karagiannidis, ``Statistical age of information: A risk-aware metric and its applications in status updates,'' \emph{IEEE Trans. Wireless Commun.}, vol.~24, no.~3, pp. 2325--2340, Mar.~2025.

\bibitem{fas1}
K.-K.~Wong, K.-F.~Tong, Y.~Shen, Y.~Chen, and Y.~Zhang, ``{Bruce Lee}-inspired fluid antenna system: Six research topics and the potentials for {6G},'' \emph{Frontiers Commun. Netw.}, vol.~3, Mar.~2022, Art.~no.~853416.

\bibitem{fas2}
A.~Shojaeifard {\em et al.}, ``{MIMO} evolution beyond {5G} through reconfigurable intelligent surfaces and fluid antenna systems,'' \emph{Proc. IEEE}, vol.~110, no.~9, pp. 1244--1265, Sep.~2022.

\bibitem{fas3}
J.~Zheng {\em et al.}, ``Flexible-position {MIMO} for wireless communications: Fundamentals, challenges, and future directions,'' \emph{IEEE Wireless Commun.}, vol.~31, no.~5, pp. 18--26, Oct.~2024.



\bibitem{fas4}
K. K. Wong, K. F. Tong, Y. Zhang, and Z. Zheng, ``Fluid antenna system for 6G: When Bruce Lee inspires wireless communications,'' {\em IET Elect. Lett.}, vol. 56, no. 24, pp. 1288--1290, Nov. 2020.

\bibitem{fas}
K.-K. Wong, A.~Shojaeifard, K.-F. Tong, and Y.~Zhang, ``Fluid antenna systems,'' \emph{IEEE Trans. Wireless Commun.}, vol.~20, no.~3, pp. 1950--1962, Mar. 2021.

\bibitem{new2025fluid}
W.~K.~New \emph{et al.},
``Fluid antenna systems: Redefining reconfigurable wireless communications,''
\emph{IEEE J. Sel. Areas Commun.},
vol.~44, pp.~1013--1044, Nov.~2025.

\bibitem{RIS}
F.~R.~Ghadi, K.-K.~Wong, W.~K.~New, H.~Xu, R.~Murch, and Y.~Zhang, ``On performance of {RIS}-aided fluid antenna systems,'' \emph{IEEE Wireless Commun. Lett.}, vol.~13, no.~8, pp. 2175--2179, Aug.~2024.

\bibitem{RIS1}
J.~Yao, X.~Lai, K.~Zhi, T.~Wu, M.~Jin, C.~Pan, M.~Elkashlan, C.~Yuen, and K.-K.~Wong, ``A framework of {FAS-RIS} systems: Performance analysis and throughput optimization,'' \emph{IEEE Trans. Wireless Commun.}, vol.~25, pp. 1333--1348, 2026.

\bibitem{RIS2}
X.~Zhu, K.-K.~Wong, B.~Tang, W.~Chen, and C.-B.~Chae, ``Fluid reconfigurable intelligent surface ({FRIS}) enabling secure wireless communications,'' \emph{IEEE Wireless Commun. Lett.}, vol.~15, pp. 2408--2412, 2026.

\bibitem{NOMA1}
L.~Tlebaldiyeva, S.~Arzykulov, T.~A.~Tsiftsis, and G.~Nauryzbayev, ``Full-duplex cooperative {NOMA}-based {mmWave} networks with fluid antenna system ({FAS}) receivers,'' in \emph{Proc. Int. Balkan Conf. Commun. Netw. (BalkanCom)}, Istanbul, Turkey, Jun.~2023, pp.~1--6.

\bibitem{NOMA2}
W.~K.~New, K.-K.~Wong, H.~Xu, K.-F.~Tong, C.-B.~Chae, and Y.~Zhang, ``Fluid antenna system enhancing orthogonal and non-orthogonal multiple access,'' \emph{IEEE Commun. Lett.}, vol.~28, no.~1, pp. 218--222, Jan.~2024.

\bibitem{NOMA3}
H.~Pang, F.~Ji, M.~Wen, and Z.~Ding, ``Power minimization for fluid antenna-assisted downlink non-orthogonal multiple access systems,'' \emph{IEEE J. Sel. Topics Signal Process.}, 2026, early access, doi: 10.1109/JSTSP.2026.3678746.


\bibitem{NOMA4}
J.~Xie, Q.~Cui, Y.~Liu, Y.~Chen, X.~Yue, and X.~Tao, ``Exploiting fluid antenna system in {NOMA} satellite communication networks,'' \emph{IEEE Trans. Wireless Commun.}, vol.~25, pp. 15183--15198, 2026.


\bibitem{hong2026fluid}
H.~Hong \emph{et al.},
``Fluid antenna multiple access for 6G: A holistic review,''
\emph{IEEE Open J. Commun. Soc.},
vol.~7, pp.~2607--2633, Mar.~2026.

\bibitem{tutorial}
Z.~Wang \emph{et al.}, ``A tutorial on extremely large-scale {MIMO} for {6G}: Fundamentals, signal processing, and applications,'' \emph{IEEE Commun. Surveys Tuts.}, vol.~26, no.~3, pp. 1560--1605, Thirdquarter 2024.

\bibitem{new2024tutorial}
W.~K.~New {\em et al.}, ``A tutorial on fluid antenna system for 6G networks: Encompassing communication theory, optimization methods and hardware designs,'' \emph{IEEE Commun. Surveys Tuts.}, vol.~27, no.~4, pp. 2325--2377, 2024.





\bibitem{mech1}
L.~Zhu, W.~Ma, and R.~Zhang, ``Movable antennas for wireless communication: Opportunities and challenges,'' \emph{IEEE Commun. Mag.}, vol.~62, no.~6, pp. 114--120, Jun.~2024.

\bibitem{mech2}
A.~Zhuravlev, V.~Razevig, S.~Ivashov, A.~Bugaev, and M.~Chizh, ``Experimental simulation of multi-static radar with a pair of separated movable antennas,'' in \emph{Proc. IEEE Int. Conf. Microw., Commun., Antennas Electron. Syst. (COMCAS)}, Tel Aviv, Israel, Nov.~2015, pp.~1--5.

\bibitem{mech3}
X.~Li, Y.~Zhou, Z.~Shen, B.~Song, and S.~Li, ``Using a moving antenna to improve {GNSS/INS} integration performance under low-dynamic scenarios,'' \emph{IEEE Trans. Intell. Transp. Syst.}, vol.~23, no.~10, pp. 17717--17728, Oct.~2022.

\bibitem{mech4}
X.~Shao and R.~Zhang, ``{6DMA} enhanced wireless network with flexible antenna position and rotation: Opportunities and challenges,'' \emph{IEEE Commun. Mag.}, vol.~63, no.~4, pp. 121--128, Apr.~2025.

\bibitem{liq1}
G.~J. Hayes, J.-H.~So, A.~Qusba, M.~D. Dickey, and G.~Lazzi, ``Flexible liquid metal alloy ({EGaIn}) microstrip patch antenna,'' \emph{IEEE Trans. Antennas Propag.}, vol.~60, no.~5, pp. 2151--2156, May~2012.

\bibitem{liq2}
A.~M. Morishita, C.~K.~Y. Kitamura, A.~T. Ohta, and W.~A. Shiroma, ``A liquid-metal monopole array with tunable frequency, gain, and beam steering,'' \emph{IEEE Antennas Wireless Propag. Lett.}, vol.~12, pp. 1388--1391, 2013.


\bibitem{liq3}
Y.~Huang, L.~Xing, C.~Song, S.~Wang, and F.~Elhouni, ``Liquid antennas: Past, present and future,'' \emph{IEEE Open J. Antennas Propag.}, vol.~2, pp. 473--487, 2021.





\bibitem{pixel1}
J.~Zhang, J.~Rao, Z.~Li, Z.~Ming, C.-Y.~Chiu, K.-K.~Wong, {\em et al.}, ``A novel pixel-based reconfigurable antenna applied in fluid antenna systems with high switching speed,'' \emph{IEEE Open J. Antennas Propag.}, vol.~6, no.~1, pp. 212--228, 2024.


\bibitem{pixel2}
B.~Liu, T.~Wu, K.~K.~Wong, H.~Wong, and K.~F.~Tong,
``Wideband pixel-based fluid antenna system: An antenna design for smart city,''
\emph{IEEE Internet Things J.},
vol.~13, no.~4, pp.~6850--6862, Nov.~2025.

\bibitem{pixel3}
K.-K.~Wong, C.~Wang, S.~Shen, C.-B.~Chae, and R.~Murch, ``Reconfigurable pixel antennas meet fluid antenna systems: A paradigm shift to electromagnetic signal and information processing,'' \emph{IEEE Wireless Commun.}, vol.~33, no.~1, pp. 191--198, Feb.~2026.



\bibitem{Tong2025}
K.~F.~Tong, B.~Liu, and K.~K.~Wong,
``Designs and challenges in fluid antenna system hardware,''
\emph{Electronics},
vol.~14, no.~7, Art.~no.~1458, Apr.~2025.


\bibitem{pro1}
Y.~Shen \emph{et al.},
``Design and experimental validation of mmWave surface-wave-enabled fluid antennas for future wireless communications,''
\emph{IEEE Antennas Wireless Propag. Lett.},
vol.~25, no.~4, pp.~1467--1471, Apr.~2026.


\bibitem{pro2}
J.~O.~Martinez, B.~G.~Guzman, and A.~G.~Armada, ``Correlation and drop velocity in fluid antenna systems: Modeling and performance,'' \emph{IEEE J. Sel. Areas Commun.}, vol.~44, pp. 1322--1334, 2026.


\bibitem{pro3}
R.~Wang {\em et al.}, ``Electromagnetically reconfigurable fluid antenna system for wireless communications: Design, modeling, algorithm, fabrication, and experiment,'' \emph{IEEE J. Sel. Areas Commun.}, vol.~44, pp. 1464--1479, 2026.

\bibitem{pro4}
S.~Zhang, Y.~Zhang, H.~Hashida, Y.~C.~Eldar, M.~Di~Renzo, and B.~Di, ``Fluid antenna systems enabled by reconfigurable holographic surfaces: Beamforming design and experimental validation,'' \emph{IEEE J. Sel. Areas Commun.}, vol.~44, pp. 1417--1431, 2026.







\bibitem{FAS_ER}
K.~K. Wong, A.~Shojaeifard, K.-F. Tong, and Y.~Zhang, ``Performance limits of fluid antenna systems,'' \emph{IEEE Commun. Lett.}, vol.~24, no.~11, pp. 2469--2472, Nov. 2020.

\bibitem{nakagami1}
L.~Tlebaldiyeva, G.~Nauryzbayev, S.~Arzykulov, A.~Eltawil, and T.~Tsiftsis, ``Enhancing {QoS} through fluid antenna systems over correlated Nakagami-$m$ fading channels,'' in \emph{Proc. IEEE Wireless Commun. Netw. Conf. (WCNC)}, Austin, TX, USA, Apr. 2022, pp. 78--83.
\bibitem{nakagami2}
J.~D. Vega-Sánchez, L.~Urquiza-Aguiar, M.~C.~P. Paredes, and D.~P.~M. Osorio, ``A simple method for the performance analysis of fluid antenna systems under correlated Nakagami-$m$ fading,'' \emph{IEEE Wireless Commun. Lett.}, vol.~13, no.~2, pp. 377--381, Feb. 2024.
\bibitem{rician}
J.~Huangfu {\em et al.}, ``Performance analysis of fluid antenna system under spatially-correlated {Rician} fading channels,'' \emph{IEEE Trans. Wireless Commun.}, vol.~25, pp. 1394--1407, 2026.
\bibitem{Alvim2024}
P.~D.~Alvim {\em et al.}, ``On the performance of fluid antennas systems under {$\alpha$-$\mu$} fading channels,'' \emph{IEEE Wireless Commun. Lett.}, vol.~13, no.~1, pp. 108--112, Jan.~2024.

\bibitem{Beaulieu2011}
N.~C.~Beaulieu and K.~T.~Hemachandra, ``Novel simple representations for {Gaussian} class multivariate distributions with generalized correlation,'' \emph{IEEE Trans. Inf. Theory}, vol.~57, no.~12, pp. 8072--8083, Dec.~2011.

\bibitem{Khammassi2023}
M.~Khammassi, A.~Kammoun, and M.-S.~Alouini, ``A new analytical approximation of the fluid antenna system channel,'' \emph{IEEE Trans. Wireless Commun.}, vol.~22, no.~12, pp. 8843--8858, Dec.~2023.

\bibitem{copula1}
F.~Rostami~Ghadi, K.-K.~Wong, F.~J.~L\'{o}pez-Mart\'{i}nez, and K.-F.~Tong, ``Copula-based performance analysis for fluid antenna systems under arbitrary fading channels,'' \emph{IEEE Commun. Lett.}, vol.~27, no.~11, pp. 3068--3072, Nov.~2023.
\bibitem{copula2}
R.~Xu, Y.~Ye, X.~Chu, G.~Lu, F.~Rostami~Ghadi, and K.-K.~Wong, "{Gaussian} copula-based outage performance analysis of fluid antenna systems: Channel coefficient- or envelope-level correlation matrix?," \emph{IEEE Wireless Commun. Lett.}, vol.~15, pp. 465--469, 2026.
\bibitem{copula3}
Y.~Hou {\em et al.}, ``A copula-based approach to performance analysis of fluid antenna system with multiple fixed transmit antennas,'' \emph{IEEE Wireless Commun. Lett.}, vol.~13, no.~2, pp. 501--504, Feb.~2024.

\bibitem{block}
P.~Ram\'{i}rez-Espinosa, D.~Morales-Jimenez, and K.-K.~Wong, ``A new spatial block-correlation model for fluid antenna systems,'' \emph{IEEE Trans. Wireless Commun.}, vol.~23, no.~11, pp. 15829--15843, Nov.~2024.

\bibitem{block1}
X.~Lai, T.~Wu, and L.~Mai, ``A variable block-correlation model for fluid antenna systems,'' in \emph{Proc. IEEE 26th Int. Workshop Signal Process. Artif. Intell. Wireless Commun. (SPAWC)}, Surrey, United Kingdom, 2025.
\bibitem{block2}
X.~Lai {\em et al.}, ``Revisiting spatial block-correlation model for fluid antenna systems: From constant to variable correlations,'' \emph{IEEE J. Sel. Areas Commun.}, vol.~44, pp. 1335--1351, 2026.
\bibitem{block3}
T.~Wu {\em et al.}, ``Variable block-correlation modeling and optimization for secrecy analysis in fluid antenna systems,'' \emph{IEEE Trans. Wireless Commun.}, vol.~25, pp. 15069--15085, 2026.


\bibitem{marcumq}
M.~K. Simon, \emph{Probability distributions involving Gaussian random variables: A handbook for engineers and scientists}. Springer, 2002.

\bibitem{simon2005digital}
M.~K. Simon and M.-S. Alouini, \emph{Digital Communication over Fading Channels}. 2nd ed. New York, NY, USA: John Wiley \& Sons, 2005.
\bibitem{integral}
I.~S. Gradshteyn and I.~M. Ryzhik, \emph{Table of Integrals, Series and Products}. New York: Academic Press, 6th ed, 2000.

\bibitem{Alouini1999}
M.-S.~Alouini and A.~Goldsmith, ``Capacity of {Rayleigh} fading channels under different adaptive transmission and diversity-combining techniques,'' \emph{IEEE Trans. Veh. Technol.}, vol.~48, no.~4, pp. 1165--1181, Jul.~1999.


\bibitem{wang2003simple}
Z.~Wang and G.~B. Giannakis, ``A simple and general parameterization quantifying performance in fading channels,'' \emph{IEEE Trans. Commun.}, vol.~51, no.~8, pp. 1389--1398, Aug. 2003.




\bibitem{abramowitz2006handbook}
M.~Abramowitz and I.~A. Stegun, \emph{Handbook of Mathematical Functions: With Formulas, Graphs, and Mathematical Tables}. Washington, DC, USA: National Bureau of Standards, 2006.

\end{thebibliography}


\end{document}